\documentclass[A4paper,USenglish]{article}
\usepackage{authblk,a4wide,url}
\usepackage{amsmath,amssymb,amsthm}
\usepackage{mathtools,mathrsfs,complexity}
\usepackage{enumerate}
\usepackage[numbers,sort&compress]{natbib}
\usepackage{relsize}
\usepackage{tikz}
\usetikzlibrary{shapes.geometric, automata,%
   positioning}

\newtheorem{theorem}{Theorem}
\newtheorem{lemma}[theorem]{Lemma}
\newtheorem{corollary}[theorem]{Corollary}
\newcommand{\bR}{\mathbb{R}}

\newcommand{\bN}{\mathbb{N}}
\newcommand{\bQ}{\mathbb{Q}}
\newcommand{\setA}{\mathscr{A}}
\newcommand{\setB}{\mathscr{B}}
\newcommand{\setC}{\mathscr{C}}
\newcommand{\setD}{\mathscr{D}}
\newcommand{\setE}{\mathscr{E}}
\newcommand{\setR}{\mathscr{R}}
\newcommand{\1}{\mathbf{1}}

\newcommand{\ent}{H}

\newcommand{\arena}{\mathbf{A}}
\newcommand{\forest}{\mathbf{F}}
\newcommand{\pref}{{\mathop{\mathtt{pref}}}}
\newcommand{\meanp}{\mathop{\mathtt{mp}}}
\newcommand{\minimax}{\mathop{\mathtt{mm}}}
\newcommand{\conv}{\mathop{\mathtt{conv}}}

\newcommand{\transpose}{\mathsmaller{\mathsf{T}}}

\let\ge\geqslant
\let\le\leqslant
\let\geq\geqslant
\let\leq\leqslant

\newcommand{\MID}[0]{\;\middle|\;} 
\newcommand{\Mid}[0]{\;\big|\;}    

\title{Entropy Games  and Matrix Multiplication Games\footnote{The support of Agence Nationale de la Recherche under the project EQINOCS (ANR-11-BS02-004) is gratefully acknowledged.
The results of Section~\ref{sec:koz} were obtained at the Institute for Information Transmission Problems,  Russian Academy of Science, by V.~Kozyakin at the expense of the Russian Science Foundation  (project 14-50-00150).}}
\author[1]{Eugene Asarin}
\author[2]{Julien Cervelle}
\author[1]{Aldric Degorre}
\author[2]{C\u at\u alin Dima}
\author[1]{Florian Horn}
\author[3]{Victor Kozyakin}
\affil[1]{LIAFA, University Paris Diderot and CNRS, France}
\affil[2]{LACL, University Paris-Est Cr\'eteil, France}
\affil[3]{IITP, Russian Academy of Science, Russia}

\begin{document}
\maketitle
\begin{abstract} Two intimately related  new classes of games are introduced and studied: entropy games (EGs) and matrix multiplication games (MMGs).
An EG is played on a finite arena by two-and-a-half players: Despot, Tribune and the non-deterministic People.  Despot wants to make the set of possible People's behaviors as small as possible, while Tribune wants to make it  as large as possible.
An MMG  is played by two players that alternately write matrices from some predefined finite sets. One wants to maximize the growth rate of the product, and the other to minimize it. We show that in general MMGs are undecidable in quite a strong sense.
On the positive side,  EGs correspond to a subclass of MMGs,  and we  prove that such MMGs and EGs are determined, and that the optimal strategies are simple. The complexity of solving such games is in $\NP\cap\coNP$.
\end{abstract}
\section{Introduction}
In recent years, some of us have been  working on a new non-probabilistic quantitative approach
to classical models in computer science based on the notion of language entropy (growth rate).
This approach has produced new insights about timed automata and languages \cite{AsaBD15:IC}
as well as temporal logics \cite{AsaBDDM14:CSL}.
In this article, we apply it to game theory and obtain a new natural class of games that we call \emph{entropy games} (EGs).
Such a game is played on  a finite arena in a turn-based way, in infinite time, by two-and-a-half\footnote{Although this term is mostly used for stochastic games, it is also an appropriate description of EGs.} players:
\emph{Despot}, \emph{Tribune} and the non-deterministic \emph{People}.
Whenever Despot and Tribune decide on their strategies $\sigma$ and $\tau$,
it leaves a set $L(\sigma,\tau)$ (an $\omega$-language) of possible behaviors of People.
Despot wants $L(\sigma,\tau)$ to be as small as possible, while Tribune wants to make
this language as large as possible.
Formally the payoff of the game is the entropy  of $L(\sigma,\tau)$,
with Despot minimizing and Tribune maximizing this value. 

Potentially these games can be used to model hidden channel capacity problems in computer security, 
where the aim of the security policy (Despot) is to minimize the information flow whatever the environment (Tribune) does. 
EGs can also be rephrased in terms of population dynamics, where one player aims to maximize the population growth rate, 
while the other minimizes it; applications of this setting to medicine, ecology, and computer security (virus propagation) are still to be explored. 
On the theoretical side, well-known mean-payoff games on finite graphs can be seen as a subclass of our EGs.  
However the purpose of this paper is to explore the theoretical setting of EGs, we therefore leave 
applications and identification of relevant subclasses of EGs for further work.

The second class of objects  studied is  that of \emph{matrix multiplication games} (MMGs), 
which came naturally when analyzing EGs and is, in our opinion, novel and interesting on its own. 
In such a game,  two players, Adam and Eve, each possess a set of matrices,  $\setA$ and $\setE$, respectively. 
The game is  played in a turn-based way, in infinite time. At every turn, the player writes a matrix from his or her set. 
Adam wants the norm of the product of matrices $A_1E_1A_2E_2\dots$ obtained to be as small  as possible (in the limit), 
while Eve wants it to be as large as possible. Formally, the payoff is the growth rate of the norm of the product.

The main interest of MMGs comes from the observation that, in the case when one of the two players is trivial (i.e.~his or her set contains only the identity matrix), the game turns into the classical, important,  and difficult, problem of computing the joint spectral radius or the joint spectral subradius of a set of matrices, see \cite{Blondel-TsitsiklisB97,Jungers:09}. Thus, MMGs is a game (or alternating) generalization of this problem. It is thus unsurprising that, in the general case, MMGs are even more difficult to analyze. We prove that several natural problems for MMGs are undecidable, in particular it is impossible to distinguish between games with value 0 and 1 (and thus it is impossible to approximate the value of an MMG).

Fortunately, MMGs have tractable subclasses. We reduce EGs to a particular subclass of MMGs (referred to as IMMGs), when the sets $\setA$ and $\setE$ are so-called \emph{independent row uncertainty sets} of non-negative matrices \cite{BN:SIAMJMAA09}, and show that for this class the game can be solved: it is determined, and for each player the optimal strategy is to write one and the same matrix at every turn. This result is based on a new, quite technical, \emph{minimax} theorem on the spectral radius of products of the type $AB$ where both $A$ and $B$ belong to sets of matrices with independent row uncertainties. We deduce that  EGs are determined, and that the optimal strategies for Despot and Tribune are positional. A careful complexity analysis of the games considered (EGs and IMMGs) allows to  prove that comparing their value to a rational constant can be done with complexity $\NP\cap\coNP$.

The article is structured as follows.
In Section~\ref{sec:prelim} we recall useful notions from linear algebra and language theory.
In Section~\ref{sec:2games} we formally define the two games and show how they are related, we also prove undecidability of general MMGs.
In Section~\ref{sec:koz} we prove the key technical minimax theorem for matrices.
In Section~\ref{sec:solving} we prove the main properties of EGs and IMMGs: determinacy, existence of simple strategies and complexity bounds.
In Section~\ref{sec:related} we  relate the EGs studied here to classical mean-payoff games
and a new kind of population games.
We conclude with a discussion on the perspectives. 
The Appendix contains proofs of all lemmas.

\section{Preliminaries}\label{sec:prelim}
\subsection{Some Linear Algebra}
Given two vectors $x, y \in \bR^{N}$, we write $x \ge y$, if
$x_i \ge y_i $ for each $1\le i\le N$.
Similar notation will be applied to matrices.
We denote by $\|\cdot\|$ the $1$-norm of vectors and matrices.
Note that, for non-negative vectors and matrices, $\|x\| = \sum_i x_i$.

Let $A$ be an  $(N\times N)$-matrix.
Its \emph{spectral radius} is defined as the maximal modulus of its eigenvalues and denoted by $\rho(A)$.
It characterizes the growth rate of $A^n$ for $n\to\infty$: according to Gelfand's formula
$
\rho(A)=\lim_{n\to\infty}\|A^n\|^{1/n}.
$
The spectral radius depends continuously on the matrix, and is monotone for non-negative matrices \cite[Corollary~8.1.19]{HJ2:e}:
$\rho(A)\le\rho(B)$ when $0\le A\le B$.
If $A> 0$, i.e.~all the elements of $A$ are positive, then by
the Perron-Frobenius theorem, the number $\rho(A)$ is a simple eigenvalue of the
matrix $A$, and all the other eigenvalues of $A$ are strictly less than
$\rho(A)$ in modulus. The eigenvector $v =(v_{1},v_{2},\ldots,v_{N})^{\transpose}$
corresponding to the eigenvalue $\rho(A)$ (normalized, for example, by the
equation $\sum v_i = 1$) is uniquely determined and
positive.

Following \cite{BN:SIAMJMAA09}, given $N$ sets of $M$-dimensional rows $\setA_{i}$ we define the \emph{IRU-set}
(independent row uncertainty set)  $\setA$ of $(N\times M)$-matrices
that
consists of all matrices of the form $A = (a_{ij})_{\stackrel{1\le i \le N}{\scriptscriptstyle 1\le j \le M}}$
wherein each of the rows $a_{i} = [a_{i1}, a_{i2}, \ldots, a_{iM}]$ belongs
to the respective  $\setA_{i}$.
We will need several simple properties of IRU-sets.
\begin{lemma}\label{lem:iru}
For an IRU-set $\setA$  formed by sets of rows
$\setA_{1},\setA_{2},\ldots,\setA_{N}$ the following holds:
\begin{enumerate}[\rm(i)]
\item	for any matrix $B$ the set $\setA B=\{AB\Mid A\in\setA\}$ is IRU as well;
\item the convex hull $\conv(\setA)$
is the IRU-set formed by the row sets $\conv(\setA_{1}), \ldots, \conv(\setA_{N})$;
\item the set $\setA$ is
 compact if and only if so are all the row  sets $\setA_{1}, \setA_{2},\ldots, \setA_{N}$.
\end{enumerate}
\end{lemma}

\subsection{Joint Spectral Radius and Subradius}

The \emph{joint spectral radius} \cite{RotaStr:IM60,DaubLag:LAA92,DaubLag:LAA01}
of a bounded set $\setA$ of $(N\times N)$-matrices  characterizes the maximal growth rate
of products of $n$ matrices from the set and admits the following equivalent definitions
(where the identity between the upper and the lower formulas constitutes the famous Berger-Wang Theorem~\cite{BerWang:LAA92}):
\begin{multline}\label{E-GSRad}
\hat{\rho}(\setA)=\lim_{n\to\infty}\sup\left\{\|A_{1}\cdots A_{n}\|^{1/n}\MID A_{i}\in\setA\right\}
=\inf_{n\ge1}\sup\left\{\|A_{1}\cdots A_{n}\|^{1/n}\MID A_{i}\in\setA\right\}\\
=\lim_{n\to\infty}\sup\left\{\rho(A_{1}\cdots A_{n})^{1/n}\MID A_{i}\in\setA\right\}
=\sup_{n\ge1}\sup\left\{\rho(A_{1}\cdots A_{n})^{1/n}\MID A_{i}\in\setA\right\}.
\end{multline}
For a compact (closed and bounded) set  $\setA$, the suprema in
\eqref{E-GSRad} may be replaced by maxima.

The \emph{joint spectral subradius} \cite{Gurv:LAA95}, or \emph{lower spectral radius}, corresponds to the minimal growth rate of products of matrices:
\begin{multline*}\label{E-LSRad0}
\check{\rho}(\setA)=
\lim_{n\to\infty}\inf\left\{\|A_{1}\cdots A_{n}\|^{1/n}\MID A_{i}\in\setA\right\}
=\inf_{n\ge1}\inf\left\{\|A_{1}\cdots A_{n}\|^{1/n}\MID A_{i}\in\setA\right\}
\\
=\lim_{n\to\infty}\inf\left\{\rho(A_{1}\cdots A_{n})^{1/n}\MID A_{i}\in\setA\right\}
=\inf_{n\ge1}\inf\left\{\rho(A_{1}\cdots A_{n})^{1/n}\MID A_{i}\in\setA\right\}.
\end{multline*}
The equivalence of the characterizations based on norms and on spectral radii is established in~\cite[Theorem~B1]{Gurv:LAA95} for finite sets $\setA$,
and in~\cite[Lemma~1.12]{Theys:PhD05} and~\cite[Theorem~1]{Czornik:LAA05} for arbitrary sets $\setA$.
Calculating the joint and lower spectral radii is a challenging problem,
and only in exceptional cases these characteristics may be found explicitly, see, e.g., \cite{Jungers:09,Koz:IITP13} and the bibliography therein.
%
The case of compact IRU-sets of non-negative matrices is such an exception, for which $\hat{\rho}$ and $\check{\rho}$  admit a simple characterization: as stated in \cite[Theorem~2]{NesPro:SIAMJMAA13}, for such a set  $\setA$ the following equalities hold:
\begin{equation}\label{E:JSR-LSR}
\hat{\rho}(\setA)= \max_{A \in \setA} \rho(A),\quad
\check{\rho}(\setA)=\min_{A \in \setA} \rho(A).
\end{equation}
Compact IRU-sets of non-negative matrices  and their convex hulls have another useful property:
as is shown in \cite[Corollary~1]{NesPro:SIAMJMAA13},
\begin{equation}\label{E:minmax-conv}
\max_{A \in \setA} \rho(A)=\max_{A \in \conv(\setA)} \rho(A),\quad
\min_{A \in \setA} \rho(A)=\min_{A \in \conv(\setA)} \rho(A),
\end{equation}
and hence $\hat{\rho}(\setA)=\hat{\rho}(\conv(\setA))$,  $\check{\rho}(\setA)=\check{\rho}(\conv(\setA))$.

\subsection{Entropy of an $\omega$-Language}
The notion of entropy of a language and methods for computing it in the case of regular languages were introduced in \cite{Chomsky58}
for finite words and in \cite{staigerentropy} for infinite ones. We will use the latter definition.
The entropy of an $\omega$-language $L \subseteq \Sigma^{\omega}$ is defined as
\[
\ent(L) = \limsup_{n \to \infty} \frac{\log |\pref_n(L)|}{n}
\]
(all the logarithms here are in base $2$), where $\pref_n(L)$ is the set of prefixes of length $n$ of infinite words in  $L$.
Intuitively, $\ent(L)$ is
the information content (``bandwidth''), measured in  bits per symbol, in typical words of the language.
In particular,  $\ent(\Sigma^{\omega}) = \log|\Sigma|$.

For a regular  $L \subseteq \Sigma^\omega$ accepted by a given B\"uchi automaton,
its entropy can be effectively computed as follows:
compute the (finite) automaton  recognizing $\pref(L)$, determinize it, and compute the entropy as the logarithm of the spectral radius of the adjacency matrix of the automaton obtained.

\section{The Two Games}\label{sec:2games}

\subsection{Entropy Games}

Consider the arena  $(D,T,\Sigma,\Delta)$ where $D$ and $T$ are disjoint finite sets of vertices (of two players),
$\Sigma$ a finite alphabet of actions and  $\Delta\subseteq  T\times \Sigma\times D \cup  D\times \Sigma\times T$  is a transition relation.
Given such an arena, we define  a game with two-and-a-half players: Despot, Tribune and People. The latter plays non-deterministically and  counts for half a player.
People chooses the initial state in $D$.
When the game is in a state $d$ of $D$, Despot plays an action $a\in  \Sigma$ and the game changes  to  some $t\in T$ (chosen by People) such that $(d,a,t)\in \Delta$.
Then, Tribune plays an action $b\in  \Sigma$ and the game changes its state to $d'\in D$, again chosen by People and such that $(t,b,d')\in \Delta$.
It is again Despot's turn.
The players must not block the game: they always choose an action that has
a corresponding transition  $(d,a,\cdot)\in \Delta$, or  $(t,b,\cdot)\in\Delta$, respectively.
We assume that the arena is non-blocking: at every state there is at least one such transition.
Figure~\ref{fig:EG-example} shows an example of such an arena, which we will use as a running example in this paper.
\newcommand{\hyperedge}[7][black]{
      \coordinate (mid#3#4) at (barycentric cs:#3=1,#4=1);
      \coordinate (first#2#3#4) at (barycentric cs:#2=1,mid#3#4=1) (first#2#3#4) node[#7,color=#1] {#5};
      \coordinate (control#2#3#4) at (barycentric cs:#2=1,mid#3#4=2);
     \draw[relative,color=#1]
  (#2)  edge[out=#6,in=0] (first#2#3#4);
  \draw[->,out=#6,in=#6+180,color=#1] (first#2#3#4)  ..controls(control#2#3#4).. (#3);
 \draw[->,out=#6,in=#6+180,color=#1] (first#2#3#4) ..controls(control#2#3#4)..  (#4);
}

\newcommand{\hyperedgea}[6][black]{\hyperedge[#1]{#2}{#3}{#4}{#5}{#6}{above}}
\newcommand{\hyperedgeb}[6][black]{\hyperedge[#1]{#2}{#3}{#4}{#5}{#6}{below}}

\newcommand{\edgea}[5][black]{
     \draw[color=#1]   (#2)  edge[->,bend left=#5] node[above] {#4} (#3);
}
\newcommand{\edgeb}[5][black]{
     \draw[color=#1]   (#2)  edge[->,bend right=#5] node[below] {#4} (#3);
}
\newcommand{\edgeab}[5][black]{
     \draw[color=#1]   (#2)  edge[->,bend left=#5] node[below] {#4} (#3);
}

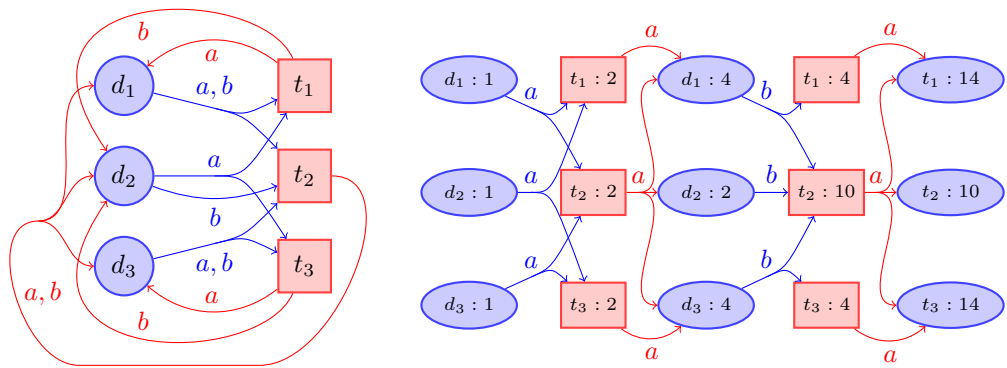
\begin{figure}[htbp!]
\begin{center}
\begin{tikzpicture}[scale=.6]
\tikzstyle{despot}=[circle,thick,draw=blue!75,fill=blue!20,minimum size=5mm]
\tikzstyle{tribune}=[rectangle,thick,draw=red!75,fill=red!20,minimum size=7mm]
\tikzstyle{people}=[regular polygon, regular polygon sides=3,fill=black!25]
\tikzstyle{desptrans}=[draw=blue!75, bend right=50]

\node[despot] (D0) at (-4,1) {$d_3$};
\node[despot] (D1) at (-4,3) {$d_2$};
\node[despot] (D2) at (-4,5) {$d_1$};
\node[tribune] (T0) at (0,1) {$t_3$};
\node[tribune] (T1) at (0,3) {$t_2$};
\node[tribune] (T2) at (0,5) {$t_1$};

\hyperedgeb[blue]{D0}{T0}{T1}{$a,b$}{0};
\hyperedgea[blue]{D1}{T0}{T2}{$a$}{0};
\hyperedgea[blue]{D2}{T1}{T2}{$a,b$}{0};
\edgeab[blue]{D1}{T1}{$b$}{-20};
\edgeb[red]{T2}{D2}{$a$}{40};
\edgea[red]{T0}{D0}{$a$}{40};
\draw[->,color=red] (T2)  ..controls (-1, 7.3) and (-7, 7.3) .. node[below] {$b$} (D1);
\draw[->,color=red] (T0)  ..controls (-1,-1.3) and (-7,-1.3) .. node[above] {$b$} (D1);
\draw[color=red]  (T1)  edge[out=0,in=0] (-1,-1.2)
 (-1,-1.2) edge[out=180,in=0] (-5,-1.2)
 (-5,-1.2)  edge[out=180,in=180]  node[right] {$a,b$} (-5.9,2)
 (-5.9,2) edge[->,out=0, in=180] (D0)
    edge[->,out=0, in=180] (D1)
    edge[->,out=0, in=180] (D2);
\end{tikzpicture}
\begin{tikzpicture}[scale=.5]
\tikzstyle{despot}=[font=\scriptsize,ellipse,thick,draw=blue!75,fill=blue!20,minimum size=3mm]
\tikzstyle{tribune}=[font=\scriptsize,rectangle,thick,draw=red!75,fill=red!20,minimum size=6mm]
\tikzstyle{people}=[regular polygon, regular polygon sides=3,fill=black!25]
\tikzstyle{desptrans}=[draw=blue!75, bend right=50]


\node[despot] (D0) at (-12,0) {$d_3: 1$};
\node[despot] (D1) at (-12,3) {$d_2: 1$};
\node[despot] (D2) at (-12,6) {$d_1: 1$};

\node[tribune] (T0a) at (-8.7,0) {$t_3: 2$};
\node[tribune] (T1a) at (-8.7,3) {$t_2: 2$};
\node[tribune] (T2a) at (-8.7,6) {$t_1: 2$};
\hyperedgea[blue]{D0}{T0a}{T1a}{$a$}{0};
\hyperedgea[blue]{D1}{T0a}{T2a}{$a$}{0};
\hyperedgea[blue]{D2}{T1a}{T2a}{$a$}{0};

\node[despot] (D0aa) at (-5.7,0) {$d_3: 4$};
\node[despot] (D1aa) at (-5.7,3) {$d_2: 2$};
\node[despot] (D2aa) at (-5.7,6) {$d_1: 4$};
\edgea[red]{T2a}{D2aa}{$a$}{40};
\edgeb[red]{T0a}{D0aa}{$a$}{40};
\coordinate (forkaa) at (barycentric cs:T1a=3,D0aa=1,D2aa=1);
\draw[color=red]  (T1a) edge  (forkaa) (forkaa) node[above] {$a$}
  (forkaa)
  edge[->,out=0, in=180] (D0aa)
  edge[->,out=0, in=180] (D1aa)
  edge[->,out=0, in=180] (D2aa);

\node[tribune] (T0aab) at (-2.5,0) {$t_3: 4$};
\node[tribune] (T1aab) at (-2.5,3) {$t_2: 10$};
\node[tribune] (T2aab) at (-2.5,6) {$t_1: 4$};
\hyperedgea[blue]{D0aa}{T0aab}{T1aab}{$b$}{0};
\hyperedgea[blue]{D2aa}{T1aab}{T2aab}{$b$}{0};
\edgea[blue]{D1aa}{T1aab}{$b$}{0};

\node[despot] (D0aaba) at (0.8,0) {$t_3: 14$};
\node[despot] (D1aaba) at (0.8,3) {$t_2: 10$};
\node[despot] (D2aaba) at (0.8,6) {$t_1: 14$};
\edgea[red]{T2aab}{D2aaba}{$a$}{40};
\edgeb[red]{T0aab}{D0aaba}{$a$}{40};
\coordinate (forkaaba) at (barycentric cs:T1aab=3,D0aaba=1,D2aaba=1);
\draw[color=red]  (T1aab) edge  (forkaaba) (forkaaba) node[above] {$a$}
  (forkaaba)
  edge[->,out=0, in=180] (D0aaba)
  edge[->,out=0, in=180] (D1aaba)
  edge[->,out=0, in=180] (D2aaba);

\end{tikzpicture}
\end{center}
\caption{\emph{Left.} Arena of our running example of an entropy game. Circles are states of the Despot while squares are states of the Tribune. At each move, the player has to choose between actions $a$ and $b$, the outcome of which may sometimes be non-deterministic (e.g. when Despot plays $a$ in state $d_2$, the next state may non-deterministically be either $t_1$ or $t_3$).
\emph{Right.} A finite play on this arena. Despot plays $ab$ (whatever his opponent does) while Tribune plays $aa$. We only give, for each step, the number of words that end up in each state controlled by the active player.}\label{fig:EG-example}
\end{figure}

A \emph{play} of the EG is a finite or infinite sequence $\pi\in (D\cdot \Sigma\cdot T\cdot \Sigma)^\infty$ compatible with the transition relation $\Delta$.
Note that four letters in a row correspond to one turn of the game.
A \emph{strategy} $\sigma$ for Despot  is a function $(D\cdot \Sigma\cdot T\cdot \Sigma)^*\cdot D\to \Sigma$ that,
given any finite play ending in a $D$ state, outputs an action taken by Despot.
The strategy is positional if it only depends on the current state of the game, i.e.~it can be expressed just as $\sigma (d)$.
A \emph{strategy} $\tau$ for Tribune is a function $(D\cdot \Sigma\cdot T\cdot \Sigma)^*\cdot D\cdot \Sigma\cdot T\to \Sigma$ which,
given any finite play ending in a $T$ state, outputs the action taken by Tribune.
The strategy is positional if it only depends on the current state of the game.
%
In a natural way we define plays compatible with a Despot's strategy $\sigma$, or with a Tribune's strategy $\tau$.
Then, given $\sigma$ and $\tau$, we have an $\omega$-language $L(\sigma,\tau)$ containing all the plays compatible with $\sigma$ and $\tau$.
In other words, $L(\sigma,\tau)$ is the set of runs that People can choose if 
Despot and Tribune commit themselves to $\sigma$ and $\tau$.
What makes EGs different from other games (parity/mean-payoff etc.) is that the payoff does not depend on a single run of the game, but on the  whole set of possible runs. More precisely, the payoff (the amount that Despot pays to Tribune) is defined as
$$
 P(\sigma,\tau)=\limsup_{n\to\infty}|\pref_{4n}(L(\sigma,\tau))|^{1/n},
$$
that is the growth rate (w.r.t.~the number of turns) of the number of plays available to the People under 
the strategies $\sigma$ and $\tau$. 
Note that the payoff is a monotone function of the entropy of $L(\sigma,\tau)$, indeed
 $$P(\sigma,\tau)=2^{4\ent(L(\sigma,\tau))},
 $$ 
 i.e.~Despot tries to diminish the entropy while Tribune aims to augment it.

\subsection{Matrix Multiplication Games}
Let $\setA$  be a set of $M\times N$-matrices and $\setE$ of $N\times M$-matrices.
The MMG between two players, Adam and Eve, is played  as follows: in turn, for every $i\in \bN$,
Adam writes a matrix  $A_i\in \setA$ and then Eve  writes a matrix  $E_i\in \setE$.
Formally, we define a \emph{play} as an infinite sequence $A_1E_1A_2E_2\dots A_iE_i\dots$ with $A_i\in\setA$ and $E_i\in \setE$.
A strategy for Adam is a function $\sigma:(\setA\cdot\setE)^*\to \setA$ that maps any finite history (which is a sequence of matrices)
to Adam's next move.
Similarly, a strategy for Eve is a mapping $\tau: (\setA\cdot\setE)^*\cdot\setA\to\setE$.
A strategy is called \emph{constant} if it does not depend on the history, i.e.~is given by just one matrix: $\sigma=A\in\setA$ or $\tau=E\in\setE$.
We define a play compatible with a strategy $\sigma$ (or  $\tau$) in a natural way.
Note that, given a strategy $\sigma$ for Adam and a strategy $\tau$ for Eve, there exists a unique play $\pi(\sigma,\tau)$  compatible with both of them.
The payoff of a play $\pi=A_1E_1A_2E_2\dots A_iE_i\dots$ (that is, the amount that Adam pays to Eve)
is the growth rate of the norm of the infinite product of matrices:
$$
P(\pi)=P(\sigma,\tau) = \limsup_{k\to\infty} \left\|\prod_{i=1}^k A_iE_i\right\|^{1/k}.
$$

\subsection{General Matrix Multiplication Games are Undecidable} 
The difficulty of general MMGs should be compared with results on the difficulty of JSR (joint spectral radius) computation. Thus, as proved in \citep[Thm 2]{BT:SCL00}, given a finite set  $\setE$ of non-negative matrices with rational elements, it is undecidable whether $\hat\rho(\setE)\leq 1$. The decidability status of the problem  $\hat\rho(\setE)<1$ is unknown. Finally, it is immediate from the characterization~\eqref{E-GSRad} that, given a precision $\varepsilon>0$, it is possible to compute $\varepsilon$-approximation of $\hat\rho(\setE)$ (in other words $\hat\rho(\setE)$ is computable as function of $\setE$ in the sense of computable analysis, see \cite{Weih2000}).

\begin{theorem}
Given a determined 	MMG with finite sets of non-negative matrices with rational elements and $\alpha\in \bQ_+$, the decision problem for its value $V\leq\alpha$ is undecidable.
\end{theorem}
\begin{proof}
Let $\setA=\{Id\}$ (Adam is trivial) and $\setE$ be  a finite set of non-negative matrices with rational elements. The corresponding MMG is determined with value $V=\hat\rho(\setE )$ and thus the decision problem $V\leq 1$ is undecidable due to \citep[Thm 2]{BT:SCL00}, cited above. 
\end{proof}
To prove stronger undecidability results for MMGs without direct counterparts for the JSR, we need  a couple of simulation lemmas: for arbitrary matrices and for non-negative ones.
\begin{lemma}\label{l:sim}
Given a two-counter machine $M$, one can construct  two finite sets of \emph{integer} matrices $\setA$ and $\setE$ such that the corresponding MMG is determined and its value $V$ satisfies:\\
	\emph{``if $M$ halts (starting with counters containing $0$) then $V=0$, else $V=1$.''}
\end{lemma}
\begin{lemma}\label{l:sim-nonneg}
Given a two-counter machine $M$, one can construct  two finite sets of non-negative integer matrices $\setA$ and $\setE$ such that the corresponding MMG  satisfies:\\
	\emph{``if $M$ halts
	then Adam can ensure payoff $<2$, otherwise Eve can ensure  payoff $\geq 2$.''}
\end{lemma}
In both cases the construction, inspired by \cite{aldric-mean}, follows the same principle: Eve tries to  simulate the machine $M$; 
if she cheats, then Adam detects this and ``resets'' the product. 
Since the halting problem is undecidable, we obtain immediately the following two theorems.
\begin{theorem}\label{thm:undecidable}
Given a determined MMG with finite sets of matrices with integer elements
\begin{itemize}
	\item its value $V$ is not computable from the matrices;
	\item it is not computable even knowing a priori  that $V\in\{0,1\}$.
\end{itemize}
\end{theorem}
Hence the MMG value cannot be approximated and is not computable (as function of $\setA$ and $\setE$) in the sense of computable analysis.
\begin{theorem}\label{thm:undecidable1}
Given an MMG with finite sets of non-negative matrices with integer elements, it is undecidable whether the maximal payoff that Eve can ensure is $<2$.
\end{theorem}


\subsection{Relations Between the Two Kinds of Games}\label{ss:rel}
Fortunately, as will be shown below, the subclass of MMGs with IRU-sets of non-negative matrices is much easier to solve. In this section, we relate EGs to such MMGs.

Let $\arena = (D,T,\Sigma,\Delta)$ be an arena with $D=\{d_1,\dots,d_M\}$ and $T=\{t_1,\dots,t_N\}$.
We define matrix sets $\setA,\setE$ as follows.
For each Despot's vertex $d_i \in D$,  and action $a\in\Sigma$
we define the row $c_{ia}=[c_{ia,1},\dots, c_{ia,N}]$ where $c_{ia,j}=1$ if $(d_i,a,t_j)\in\Delta$ and $c_{ia,j}=0$ otherwise.
Next we define the row set $\setA_i=\{c_{ia}\neq0\Mid a\in\Sigma\}$ (non-zero rows correspond to non-blocking actions).
Row sets  $\setA_1,\dots,\setA_M$ determine an IRU-set of matrices $\setA$.
The IRU-set   $\setE$ corresponding to Tribune's actions is defined similarly.
In the running example in Figure~\ref{fig:EG-example}, for instance, the row sets are the following:
$\setA_1 = \left\{ \left[ 1,1,0 \right]\right\},
\setA_2 =\left\{\left[0,1,0\right], \left[1,0,1\right]\right\},
\setA_3 =\left\{\left[0,1,1\right]\right\},
\setE_1 =\left\{\left[0,1,0\right], \left[1,0,0\right]\right\},
\setE_2 =\left\{\left[1,1,1\right]\right\},
\setE_3 =\left\{\left[0,1,0\right], \left[0,0,1\right]\right\}$.

Note first that there is a natural bijection between the positional strategies of Despot and the set $\setA$:
 any positional strategy $\sigma:D\to\Sigma$ corresponds to the matrix $A_\sigma\in \setA$ with $i$-th row $c_{i,\sigma(d_i)}$ for Adam.
Similarly,  a positional strategy of Tribune $\tau$ corresponds to Eve's matrix $E_\tau\in \setE$.
The following lemma generalizes this observation to any type of strategies:

\begin{lemma}
\label{lemma:f1}
Let $\arena $ be an arena and $\setA,\setE$ the corresponding IRU matrix sets.
Then for every  pair of strategies $(\sigma,\tau)$  of Despot and Tribune in  the EG on  $\arena$
there exists a pair of strategies  $(\varsigma,\theta)$ of Adam and  Eve in the MMG $(\conv(\setA),\conv(\setE))$
with exactly the same payoff.
Moreover, if $\sigma$  is positional, then $\varsigma$ is constant and permanently chooses $A_\sigma$.
The case of positional $\tau$ is similar.
\end{lemma}

Note that Lemma~\ref{lemma:f1} provides a rather weak relation between two games and does not mean,
by itself, that the two games have the same value.
However, we will show later (cf.~Lemma~\ref{lemma:f2}) that  optimal \textbf{constant} strategies
in the MMG that belong to $\setA$ and $\setE$ are in bijection with optimal positional strategies in the EG.

\section{Minimax Theorem for IRU-Sets of Matrices}\label{sec:koz}

In this section, we prove the key theorem of this article.
\begin{theorem}\label{T:minimax}
Let $\setA$ be a compact IRU-set of non-negative $(N\times M)$-matrices and
$\setB$ be a compact IRU-set of non-negative  $(M\times N)$-matrices. Then
\begin{equation}\label{E:minimax}
\min_{A \in \setA} \max_{B \in \setB} \rho(AB) = \max_{B \in \setB} \min_{A
\in \setA} \rho(AB).
\end{equation}
\end{theorem}
In the rest of the article we will denote this minimax by $\minimax(\setA,\setB)$.
The study of minimax relations will be based on the following
well-known fact:
\begin{lemma}[see {\citep[Section~13.4]{von1947theory}}]\label{T:mm}
Let $f (x, y)$ be a continuous function on the product of compact spaces $X
\times Y$. Then
$$
\min_{x}\max_{y} f(x,y) \ge\max_{y}\min_{x} f(x,y).
$$
The exact equality holds if and only if there exists a saddle point, i.e.~a
point $(x_{0}, y_{0})$ satisfying the inequalities
$$
f (x_{0}, y) \le f (x_{0}, y_{0}) \le f (x, y_{0})
$$
for all $x \in X$, $y \in Y$.
\end{lemma}

We will also use two lemmas on matrices. The first one  provides spectral radius bounds and is quite standard in Perron-Frobenius theory;  as usual in this  theory it relates global characteristics of a non-negative matrix  (such as spectral radius) with its behavior on one non-negative vector. 
 \begin{lemma}\label{L:1}
Let $A$ be a non-negative $(N\times N)$-matrix; then the following properties hold:
\begin{enumerate}[\rm(i)]
\item if $Au\le\rho u$ for some vector $u>0$, then $\rho\ge0$ and $\rho(A)\le\rho$;
\item \qquad if  furthermore  $A>0$ and $Au\neq\rho u$, then $\rho(A)<\rho$;
\item if $Au\ge\rho u$ for some non-zero vector $u\ge0$ and some number $\rho\ge0$, then $\rho(A) \ge\rho$;
\item\qquad if furthermore $Au\neq\rho u$, then $\rho(A)> \rho$.
\end{enumerate}
\end{lemma}

The next lemma concerning IRU-sets of matrices is new and can be explained as follows. For an IRU-set of matrices and two vectors $u$ and $v$
we imagine that the sets $B_{l}=\{x:~ x\le v\}$ and
$B_{u}=\{x:~ v\le x\}$ form the lower and upper bulbs of an
hourglass with the neck at the point $v$. The lemma 
asserts that either all the grains $Au$  (for all matrices $A$ in the set) fill one of the bulbs, or  there
remains at least one grain in the other bulb.
Clearly this alternative does not hold for general sets of matrices.

\begin{lemma}[hourglass alternative]\label{L:alternative}
Let $\setA$ be an IRU-set of $(N\times M)$-matrices and let
$\tilde{A}u=v$ for some matrix $\tilde{A}\in\setA$ and vectors $u,v $. Then
the following holds:
\begin{enumerate}[\rm(i)]
\item either $Au\ge v$ for all $A\in\setA$ or exists a matrix
    $\bar{A}\in\setA$ such that $\bar{A}u\le v$ and $\bar{A}u\neq v$;

\item either $Au\le v$ for all $A\in\setA$ or exists a matrix
    $\bar{A}\in\setA$ such that  $\bar{A}u\ge v$ and $\bar{A}u\neq v$.
\end{enumerate}
\end{lemma}

We are ready to prove the minimax theorem.
\begin{proof}[Proof of Theorem~\ref{T:minimax}]
According to Lemma~\ref{T:mm}, the minimax equality \eqref{E:minimax} may occur if and
only if  some matrices $\tilde{A} \in \setA$ and $\tilde{B} \in \setB$
satisfy the inequalities
\begin{align}\label{E:propAB1}
\rho(\tilde{A} B) &\le\rho(\tilde{A} \tilde{B})\quad\text{ for all }B \in \setB;\\
\rho(\tilde{A}\tilde{B})&\le\rho(A\tilde{B})\quad\text{ for all }A \in \setA. \label{E:propAB2}
\end{align}

Consider first the case when all the matrices in $\setA$ and $\setB$ are \textbf{positive}.
To construct the matrices $\tilde{A} \in \setA$ and $\tilde{B} \in \setB$ we
proceed as follows. For each $B\in\setB$ let $A_{B}\in\setA$ be
a matrix that minimizes (in $A$) the quantity $\rho(AB)$.
Such a matrix $A_{B}$ exists due to compactness of the set $\setA$ and
continuity  of the function $\rho(AB)$ in $A$ and $B$. Then, for each matrix
$B\in\setB$, the relations
$$
\rho(A_{B}B)=\min_{A\in\setA}\rho(AB)\le\rho(AB)
$$
hold for all $A\in\setA$. Let $\tilde{B}$ be the matrix maximizing $\min_{A\in\setA}\rho(AB)$  over the set $\setB$, and let $\tilde{A}=A_{\tilde{B}}$.
In this case
\begin{equation}\label{E:tildeAB}
\max_{B\in\setB}\rho(A_{B}B)=\max_{B\in\setB}\min_{A\in\setA}\rho(AB)=
\min_{A\in\setA}\rho(A\tilde{B})=\rho(A_{\tilde{B}}\tilde{B})=
\rho(\tilde{A}\tilde{B}),
\end{equation}
which implies inequality \eqref{E:propAB2}
for all $A\in\setA$, and it remains to prove 
 \eqref{E:propAB1}
for all $B \in \setB$.

Let $v = (v_{1}, v_{2}, \ldots, v_{N})^\transpose$ be the positive eigenvector of the
$(N\times N)$-matrix $\tilde{A} \tilde{B}$ corresponding to the eigenvalue
$\tilde{\rho} = \rho(\tilde{A}\tilde{B})$. By denoting
$$
w = \tilde{B}v\in\bR^{M}
$$
we obtain that $\tilde{\rho}v=\tilde{A}w$. Let us show that in this case
\begin{equation}\label{E:rvaaa}
\tilde{\rho}v\le Aw\quad\text{ for all }A \in \setA.
\end{equation}
Otherwise, by
Lemma~\ref{L:alternative}(i) there would exist a matrix $\bar{A}\in\setA $ such that
$\tilde{\rho}v\ge\bar{A}w$ and $\tilde{\rho}v\neq\bar{A}w$, which
implies, by the definition  of the vector $w$, that
$\tilde{\rho}v\ge \bar{A}\tilde{B}v$ and
$\tilde{\rho}v\neq\bar{A}\tilde{B}v$. Then by Lemma~\ref{L:1}
$$
\rho(\bar{A}\tilde{B})<\tilde{\rho}=\rho(\tilde{A}\tilde{B}),
$$
which contradicts \eqref{E:propAB2}. This contradiction completes the
proof of inequality \eqref{E:rvaaa}.
Similarly, now we show that
\begin{equation}\label{E:wbbb}
w\ge Bv\quad\text{ for all }B \in \setB.
\end{equation}
 Again, assuming the contrary, by
Lemma~\ref{L:alternative}(ii) there exists a matrix $\bar{B}\in\setB$ such that
$w\le \bar{B}v$ and $w\neq \bar{B}v$. This last inequality, together
with \eqref{E:rvaaa} applied to the matrix $A_{\bar{B}}$, yields
$\tilde{\rho}v\le A_{\bar{B}}\bar{B}v$ and $\tilde{\rho}v\neq
A_{\bar{B}}\bar{B}v$. Then by Lemma~\ref{L:1}
$$
\tilde{\rho} <\rho(A_{\bar{B}}\bar{B}),
$$
which contradicts \eqref{E:tildeAB} asserting that $\tilde{\rho} =
\rho(\tilde {A} \tilde{B})$  is the maximum value of the function
$\rho(A_BB)$ over all  $B\in\setB$. This contradiction completes the proof of
 inequality \eqref{E:wbbb}.

From $\tilde{\rho}v=\tilde{A}w$ and \eqref{E:wbbb} we obtain the inequality
$\tilde{\rho}v\ge\tilde{A}Bv$ valid for all $B\in\setB$, which by
Lemma~\ref{L:1} implies the relations
$$
\rho(\tilde{A} \tilde{B}) = \tilde{\rho} \ge \rho(\tilde{A} B)
$$
valid for all $B \in \setB$, or, which is the same, inequality
\eqref{E:propAB1}.
The theorem is proved for positive matrices.

Consider now the general case of compact IRU-sets of \textbf{non-negative
matrices} $\setA$ and $\setB$. If the set $\setA$ is determined by some sets of
$M$-rows $\setA_{i}$, $i=1,2,\ldots,N$, then choose an arbitrary $\varepsilon>0$
and consider the sets of rows
$$
\setA^{(\varepsilon)}_{i}=\{a^{(\varepsilon)}\Mid
a^{(\varepsilon)}=a+\varepsilon[1,1,\ldots,1],~
a\in\setA_{i}\},
$$
where $i=1,2,\ldots,N$. In this case the IRU-set of matrices
$\setA^{(\varepsilon)}$ consists of  positive matrices
$A+\varepsilon\1$, where $A\in\setA$ and $\1$ is the matrix with all elements
equal to $1$. Define similarly the IRU-set of matrices $\setB^{(\varepsilon)}$.

By the result just proved, for each $\varepsilon>0$   the minimax equality holds
for positive matrices:
$$
\min_{A\in\setA^{(\varepsilon)}}\max_{B\in\setB^{(\varepsilon)}}\rho(AB)=
\max_{B\in\setB^{(\varepsilon)}}\min_{A\in\setA^{(\varepsilon)}}\rho(AB),
$$
which by Lemma~\ref{T:mm} is equivalent to the existence of
$\tilde{A}_{\varepsilon}\in\setA$ and $\tilde{B}_{\varepsilon}\in\setB$ such
that
\[
\rho((\tilde{A}_{\varepsilon}+\varepsilon\boldsymbol{1})(B+\varepsilon\boldsymbol{1}))\le
\rho((\tilde{A}_{\varepsilon}+\varepsilon\boldsymbol{1})(\tilde{B}_{\varepsilon}+\varepsilon\boldsymbol{1}))\le
\rho((A+\varepsilon\boldsymbol{1})(\tilde{B}_{\varepsilon}+\varepsilon\boldsymbol{1}))
\]
for all $A\in\setA$ and $B\in\setB$. Taking here $\varepsilon=\varepsilon_{n}$,
where $\{\varepsilon_{n}\}$ is an arbitrary sequence of positive numbers
converging to zero, we get
\begin{equation}\label{E:mmvare1}
\rho((\tilde{A}_{\varepsilon_{n}}+\varepsilon_{n}\boldsymbol{1})(B+\varepsilon_{n}\boldsymbol{1}))\le
\rho((\tilde{A}_{\varepsilon_{n}}+\varepsilon_{n}\boldsymbol{1})(\tilde{B}_{\varepsilon_{n}}+\varepsilon_{n}\boldsymbol{1}))\le
\rho((A+\varepsilon_{n}\boldsymbol{1})(\tilde{B}_{\varepsilon_{n}}+\varepsilon_{n}\boldsymbol{1}))
\end{equation}
for all $A\in\setA$ and $B\in\setB$.
Without loss of generality, in view of the compactness of the sets $\setA$ and $\setB$, we may assume the existence of
matrices $\tilde A$ and $\tilde B$ such that
$\tilde{A}_{\varepsilon_{n}}\to\tilde{A}\in\setA$ and
$\tilde{B}_{\varepsilon_{n}}\to\tilde{B}\in\setB$ as $n\to\infty$.
Then turning to the limit in \eqref{E:mmvare1},
we obtain the inequalities
$$
\rho(\tilde{A}B)\le \rho(\tilde{A}\tilde{B})\le \rho(A\tilde{B})
$$
for all $A\in\setA$ and $B\in\setB$, which are equivalent to \eqref{E:propAB1} and
\eqref{E:propAB2}. This concludes the proof.
\end{proof}

\begin{corollary}\label{C:minimaxconv}
For IRU-sets $\setA$ and $\setB$ of non-negative matrices it holds that
\[
\minimax(\conv(\setA),\conv(\setB))=\minimax(\setA,\setB).
\]
\end{corollary}

\section{Solving the Games}\label{sec:solving}
\subsection{Solving  Matrix Multiplication Games for IRU-Sets}

\begin{theorem}\label{th:mat-cons}
Let $\setA$ and $\setE$ be  compact IRU-sets of non-negative matrices. Then the corresponding MMG is determined, and moreover Adam and Eve possess constant optimal strategies.
\end{theorem}
\begin{proof}
Let us apply Theorem~\ref{T:minimax} to matrix sets $\setA$ and $\setE$. Define $V$, $E_0$ and $A_0$ such that
\begin{equation}\label{E:saddle}
\min_{E\in\setE}\rho(EA_0)=\max_{A\in\setA}\min_{E\in\setE}\rho(EA)=\min_{E\in\setE}\max_{A\in\setA}\rho(EA)=\max_{A\in\setA}\rho(E_0A)=V.
\end{equation}
Let Adam only play $A_0$. Take any compatible play $\pi=A_0E_1A_0E_2\cdots $ and put $C_i=A_0E_i$.  Denote $\setC=\{EA_0|E\in\setE\}$; it is an IRU-set by  Lemma~\ref{lem:iru}. The payoff $P$ for $\pi$ yields
\begin{multline*}
P=\limsup_{n\to\infty} \|A_0C_1\cdots C_{n-1}E_n\|^{1/n}
\le\limsup_{n\to\infty} (\|A_0\|\cdot\|C_1\cdots C_{n-1}\|\cdot\|E_n\|)^{1/n} \\
\le\lim_{n\to\infty} K^{\frac{2}{n}} \limsup_{n\to\infty} \|C_1\cdots C_{n-1}\|^{\frac{1}{n-1}}
\le\hat{\rho}(\setC) \stackrel1= \max_{C\in\setC}\rho(C) = \max_{E\in\setE}\rho(EA_0)\stackrel2=V,
\end{multline*}
where the constant $K$ is an upper bound for the norms of the matrices in $\setA$ and $\setE$, equality $1$ comes from the first equality \eqref{E:JSR-LSR} and equality $2$ comes from \eqref{E:saddle}.

Let Eve only play $E_0$. Take any compatible play  $\pi'=A_1E_0A_2E_0\cdots $. Let us write $D_i=A_iE_0$. Denote $\setD=\{AE_0,A\in\setA\}$; it is an IRU-set. The payoff $P'$ for $\pi'$ is such that
\[
{P'}=\limsup_{n\to\infty} \|C_1\cdots C_n\|^{1/n}
\ge\liminf_{n\to\infty} \|C_1\cdots C_n\|^{1/n}\\
\ge\check{\rho}(\setD) \stackrel1= \min_{D\in\setD}
\rho(D) = \min_{A\in\setA}\rho(AE_0)\stackrel2=V,
\]
where equality $1$ comes from the second equality \eqref{E:JSR-LSR} and equality $2$  from \eqref{E:saddle} using  the equality $\rho(EA_0)=\rho(A_0 E)$.

We have proved that Adam (by constantly playing $A_0$) can ensure payoff $\le V$ whatever Eve plays; and that Eve (by constantly playing $E_0$) can ensure payoff $\ge  V$ whatever Adam plays. This concludes the proof.
\end{proof}
\begin{corollary}\label{C:gameconv}Let $\setA$ and $\setE$ be compact IRU-sets of non-negative matrices. In the MMG on
$\conv(\setA),\conv(\setE)$, the constant optimal strategies can be chosen from sets $\setA$ and $\setE$.
\end{corollary}
This follows immediately from the proof of the theorem and Corollary~\ref{C:minimaxconv}.
\subsection{Solving Entropy Games}
In this section, we consider an EG on an arena $\arena$ and the corresponding matrix sets $\setA$ and $\setE$,  as defined in Section~\ref{ss:rel}.
\begin{lemma}
\label{lemma:f2}
	Let $(\sigma,\tau)$ be two positional strategies in the EG. Then, if corresponding constant strategies $A_\sigma$ and $E_\tau$ are optimal for their respective players in the MMG with matrix sets $\conv(\setA)$ and $\conv(\setE)$, then  so are $\sigma$ and $\tau$.
\end{lemma}
\begin{theorem}\label{th:EG}
Every EG is determined, and  Despot and Tribune possess positional optimal strategies.
\end{theorem}
\begin{proof}
From Theorem~\ref{th:mat-cons}, we know that for the MMG $(\conv(\setA),\conv(\setE))$ both Adam and Eve possess optimal strategies, which consist in constantly playing  some matrices $A$ and $E$. From Corollary~\ref{C:gameconv}, the matrices $A$ and $E$ can be chosen  from sets $\setA$ and $\setE$, respectively. Then, there exist positional strategies $\sigma$ and $\tau$ on $\arena$ such that  $A=A_\sigma$ and $E=E_\tau$. By Lemma~\ref{lemma:f2}, strategies $\sigma$ and $\tau$ are optimal in the EG.
\end{proof}

Back to the running example. Here a quick exploration of the combinations of rows shows that the matrices realizing the minimax over the two IRU-sets defined by row sets $\setA_1,\setA_2, \setA_3$ and $\setE_1,\setE_2, \setE_3$ are $A=\left[\begin{smallmatrix}1&1&0\\1&0&1\\0&1&1\end{smallmatrix}\right]$ for Adam/Despot and $E=\left[\begin{smallmatrix}1&0&0\\1&1&1\\0&0&1\end{smallmatrix}\right]$ for Eve/Tribune. These matrices describe both the optimal constant strategy of the MMG and the optimal positional strategy of the EG induced by this arena. The value of both games is the spectral radius $\rho(AE)=\rho\left(\left[\begin{smallmatrix}2&1&1\\1&0&1\\1&1&2\end{smallmatrix}\right]\right)=\left(\sqrt{17}+3\right)/2\simeq 3.562$.

\subsection{Complexity Issues}
We will analyze the complexity of solving matrix multiplication (and hence entropy) game. We  start with  necessary and sufficient conditions for inequalities on joint spectral radii and subradii of IRU-sets (recall also \eqref{E:JSR-LSR} relating them to maximal and minimal spectral radii).
\begin{lemma}\label{L:equiv}
For any compact  IRU-set of positive matrices $\setA$ and  $\alpha\in \bQ_+$ the following  equivalences hold:
\begin{gather}
\hat{\rho}(\setA)<\alpha \Leftrightarrow \exists v>0 \,\forall A\in \setA (Av<\alpha v);\label{E:equiv:lt}\\
\hat{\rho}(\setA)\le\alpha \Leftrightarrow \exists v>0 \,\forall A\in \setA (Av\le\alpha v);\label{E:equiv:le}\\
\check{\rho}(\setA)>\alpha \Leftrightarrow \exists v>0 \,\forall A\in \setA (Av>\alpha v);\label{E:equiv:gt}\\
\check{\rho}(\setA)\ge\alpha \Leftrightarrow \exists v>0 \,\forall A\in \setA (Av\ge\alpha v).\label{E:equiv:ge}
\end{gather}
If the matrices are only non-negative, the equivalences \eqref{E:equiv:lt} above and \eqref{E:equiv:ge2} below hold:
\begin{equation}\label{E:equiv:ge2}
\check{\rho}(\setA)\ge\alpha \Leftrightarrow \exists (v\ge0, v\neq0) \,\forall A\in \setA (Av\ge\alpha v).
\end{equation}
 \end{lemma}

The computational aspects of calculating the values $\hat{\rho}(\setA)$ and
$\check{\rho}(\setA)$ for IRU-sets of non-negative matrices, based on relations \eqref{E:JSR-LSR}, are discussed in
\cite{BN:SIAMJMAA09,NesPro:SIAMJMAA13,Prot:MP15}. These articles provide polynomial algorithms for approximation of the minimal and maximal spectral radii, as well as a variant of the simplex method for these problems.  In the next theorem we  prove a complexity result in a form suitable for game analysis.

\begin{theorem}\label{T:P}
Given a finite IRU-set of nonnegative matrices $\setA$ with rational elements (represented by row sets $\setA_1$, $\setA_2$, \ldots, $\setA_N$), and a number  $\alpha\in \bQ_+$, the decision problems whether $\hat{\rho}(\setA)<\alpha$ and whether $\check{\rho}(\setA)\ge\alpha$ belong to the complexity class $\P$.
Moreover, if the matrices are positive, then the decision problems $\hat{\rho}(\setA)\le\alpha$ and $\check{\rho}(\setA)>\alpha$ are also in $\P$.
\end{theorem}
\begin{proof}
The polynomial algorithms are based on the previous lemma. Consider the problem of deciding $\hat{\rho}(\setA)<\alpha$, which can be rewritten using \eqref{E:equiv:lt}  as
$\exists v>0 \,\forall A\in \setA (Av<\alpha v).
$
We will not test all the matrices $A\in \setA$ (there are exponentially many of them); instead, we will treat each row separately. The condition $\forall A\in \setA (Av<\alpha v)$ can be rewritten as a system of linear inequalities: for each $i$ and for each row $[c_1,c_2,\dots,c_N]\in\setA_i$ require that
  $$
  c_1v_1+c_2v_2+\cdots+c_Nv_N<\alpha v_i.
  $$
The condition $v>0$ can be written as $N$ inequalities $v_i>0$: one for each coordinate. Using a polynomial algorithm for linear programming  we can decide whether a solution $v$ satisfying all these linear inequalities exists.

All other decision procedures, based on~\eqref{E:equiv:le}--\eqref{E:equiv:ge2}, are similar. The condition $v\ge0, v\neq0$ can be represented as a disjunction of $N$ linear systems $v_j>0\land \bigwedge_{i=1}^N v_i\ge 0$.
\end{proof}

\begin{theorem}\label{T:NP}
Given two finite IRU-sets of nonnegative matrices $\setA$ and $\setB$ with rational elements, and a number  $\alpha\in \bQ_+$, the decision problem of whether
$
\minimax(\setA,\setB)<\alpha$ 
belongs to $\NP\cap\coNP$.
Moreover, if the matrices are positive, then the  problem of whether
$
\minimax(\setA,\setB)\le \alpha$  
is also in $\NP\cap\coNP$.
\end{theorem}
\begin{proof}
Consider the problem of deciding whether $\minimax(\setA,\setB)<\alpha$, which can be rewritten as
$$
\min_{A\in\setA}\max_{B\in\setB}\rho(BA)<\alpha,
$$
 or equivalently
$$\exists A_0\in\setA (\hat\rho(\setB A_0)<\alpha).
$$
The nondeterministic polynomial algorithm proceeds as follows:
\begin{itemize}
  \item guess non-deterministically a matrix $A_0\in \setA$;
  \item compute the representation of $\setB A_0$ as an IRU-set generated by the row sets $\setC_1,\setC_2,\dots,\setC_N$;
  \item check the inequality $\hat{\rho}(\setB A_0)<\alpha$ in polynomial time using Theorem~\ref{T:P}.
\end{itemize}
We conclude that the problem $\minimax(\setA,\setB)<\alpha$ is in $\NP$.
The complementary problem $\minimax(\setA,\setB)\ge\alpha$ is also in  $\NP$, as it can be rewritten as
$$
\max_{B\in\setB}\min_{A\in\setA}\rho(AB)\ge\alpha,
$$ 
or equivalently 
$$\exists B_0\in\setB (\check\rho(\setA B_0)\ge\alpha),
$$
and decided by a non-deterministic polynomial algorithm similarly. We conclude that the two problems belong to $\NP\cap\coNP$.

For positive matrices, the proof for the other decision problem based on the second statement of  Theorem~\ref{T:P} is similar.
\end{proof}
Our main complexity result follows immediately.
\begin{theorem} \label{C:NPcoNP}Given an EG or an MMG with finite IRU-sets of non-negative matrices with rational elements and $\alpha\in \bQ_+$, the decision problem for its value: $V<\alpha$ 
is in $\NP\cap\coNP$.
\end{theorem}

\section{Related Models}\label{sec:related}

\subsection{Weighted Entropy Games}
Up to now we have considered entropy games with \emph{simple} transitions,
but it is straightforward to add multiplicities  (weights) to them.
A \emph{weighted entropy game} is played on a \emph{weighted arena} $\arena = (D,T,\Sigma,\Delta,w)$
with a function $w:\Delta\to \bN_+$ assigning weights to transitions (informally a weight is the number of ways in which a transition can be taken).
Strategies and plays are defined as in the unweighted case.
Let $L$ be some set of (infinite) plays.
For every  $u\in\pref(L)$   we define its weight $w(u)$ as the product of weights of all the transitions taken along $u$.
We define $w_n(L)=\sum_{u\in\pref_{4n}(L)}w(u)$, and finally the payoff corresponding to strategies $\sigma$ and $\tau$ of two players is defined as:
$$
P=\limsup_{n\to\infty}\left(w_n(L(\sigma,\tau))\right)^{1/n}.
$$
Our main results on EGs (Thms~\ref{th:EG} and \ref{C:NPcoNP}) extend straightforwardly to weighted EGs.

\subsection{Mean-Payoff Games}
Well-known mean-payoff finite-state games (MPG)~\cite{EhrenfeuchtMeanPayoff} can be considered as a deterministic subclass of weighted entropy games.
A (variant of) MPG is played on arena $(D,T,\Delta,w)$ with transition relation $\Delta\subseteq D\times T\cup T\times D$ and weight function $w:\Delta\to\bN$.
The play starts in some state $d_0\in D$, and  the two players choose transitions in turn.
The resulting play is an infinite word $\gamma_{d_0}\in(D\cdot T)^\omega$.
The mean-payoff  corresponding to the play $\gamma_{d_0}=d_0,t_0,d_1,t_1,\dots$ is the limit of the average weight of transitions taken:
$$
\meanp(\gamma_{d_0}) = \limsup_{n\to\infty}\frac{1}{n}\sum_{i=1}^n (w(d_{i-1},t_{i-1})+w(t_{i-1},d_{i})).
$$
Finally, player D wants to minimize and player T to maximize the payoff
$\max_{d_0\in D} \meanp(\gamma_{d_0})$.
As proved in~\cite{EhrenfeuchtMeanPayoff}, MPGs are determined and their optimal strategies are positional.
As for complexity, \cite{ZwickMeanPayoff} shows that testing whether
the value of an  MPG  is smaller than a rational $\alpha$ is  in $\NP\cap\coNP$ and becomes polynomial for weights presented in the unary system.

An MPG $\arena=(D,T,\Delta,w)$  can be transformed into a weighted EG $\arena'=(D,T,\Sigma,\Delta',w')$ as follows.
The states of both players are the same, $\Sigma$ is large enough, and for each transition $(p,q)\in \Delta$ there is a corresponding
 transition $(p,a,q)\in\Delta'$ with some $a$ (occurring only in this transition).
Its weight is $w'(p,a,q)=2^{w(p,q)}$.
We notice that the EG obtained is deterministic: due to unique transition labels for any strategies $\sigma$ and $\tau$, the language
$L(\sigma,\tau)$ contains one play for each initial state.
Strategies and plays of both games $\arena$ and $\arena'$ are now in natural bijection and the payoff of $\arena$ equals the logarithm of the payoff of $\arena'$.

This way, we obtain the classical results that MPGs are determined and both players have optimal positional strategies.
Due to the exponential encoding of payoffs, the complexity obtained using our approach is, however, not as good as using direct algorithms, see \cite{ZwickMeanPayoff}.

\subsection{Population Dynamics}
Consider an EG with arena  $\arena = (D,T,\Sigma,\Delta)$.
It can be interpreted as the following population game between two players, Damien and Theo.
Elements of $D$ and $T$ correspond to species (forms of viruses, microorganisms, etc.).
Initially there is one (or any non-zero number of) organism(s) for each species in $D$.
At his turn Damien chooses an action $a\in\Sigma$ and applies it to each organism.
An organism of species $d$, when subject to action $a$, turns into  the set of organisms of species $\{t\Mid (d,a,t)\in\Delta\}$.
Theo  plays similarly.  The aim of Damien is to minimize the growth rate of the population, while Theo wants to maximize it.
The value of the game and the optimal (positional) strategies are the same as for the EG.

\section{Conclusions}\label{sec:concl}

We have introduced two (closely interrelated) families of games: entropy games played on finite arenas (graphs), and matrix multiplication games.
The main result is that entropy games are determined and optimal strategies are positional in EG, while MMGs
for IRU-sets of non-negative matrices are determined and optimal strategies are constant.
These results are based on a new minimax theorem on spectral radii of products of IRU-sets of matrices.
The  results obtained prove the existence of equilibria in zero-sum games
with a new type of limit payoffs, which is  neither computed on a single play of the game nor probabilistic.
On the other hand, they rely upon and generalize important results on the computability
of joint spectral radii and subradii, an important problem in switching dynamic systems.

A presumably straightforward extension would be the ``probabilization'' of our game models,
in that both Despot and Tribune would be allowed to play randomized strategies.
The minimax theorem ensures the existence of optimal pure strategies for both players.
However the entropy-based payoff of the game needs to be given a proper generalization to this probabilistic setting. We may mention that such a generalization could be seen as entropy games on  stochastic branching processes, and provide interesting links with this research domain.
Finally, both our games are turn-based games with perfect information. The first generalization to be considered is to go to concurrent games --- where perhaps some
polynomial-size memory is needed, similarly to the classic case of concurrent games played on graphs in infinite time.
The more difficult case is that of games of imperfect information:  
corresponding matrix games no longer have a simple structure (independent row uncertainty), and we conjecture that analysis of such games is non-computable.
Last but not least, potential applications sketched in the introduction should be addressed.

\bibliographystyle{plain}
\bibliography{entropy}

\begin{thebibliography}{10}

\bibitem{AsaBD15:IC}
Eugene Asarin, Nicolas Basset, and Aldric Degorre.
\newblock Entropy of regular timed languages.
\newblock {\em Inform. Comput.}, 241:142--176, 2015.

\bibitem{AsaBDDM14:CSL}
Eugene Asarin, Michel Blockelet, Aldric Degorre, C\u{a}t\u{a}lin Dima, and
  Chunyan Mu.
\newblock Asymptotic behaviour in temporal logic.
\newblock In {\em Proc. {CSL-LICS}}, pages 10:1--10:9. {ACM}, 2014.

\bibitem{BerWang:LAA92}
Marc~A. Berger and Yang Wang.
\newblock Bounded semigroups of matrices.
\newblock {\em Linear Algebra Appl.}, 166:21--27, 1992.

\bibitem{BN:SIAMJMAA09}
Vincent~D. Blondel and Yurii Nesterov.
\newblock Polynomial-time computation of the joint spectral radius for some
  sets of nonnegative matrices.
\newblock {\em SIAM J. Matrix Anal. A.}, 31(3):865--876, 2009.

\bibitem{BT:SCL00}
Vincent~D. Blondel and John~N. Tsitsiklis.
\newblock The boundedness of all products of a pair of matrices is undecidable.
\newblock {\em Syst. Control Lett.}, 41(2):135--140, 2000.

\bibitem{Chomsky58}
Noam Chomsky and George~A. Miller.
\newblock Finite state languages.
\newblock {\em Inform. Control}, 1(2):91--112, 1958.

\bibitem{Czornik:LAA05}
Adam Czornik.
\newblock On the generalized spectral subradius.
\newblock {\em Linear Algebra Appl.}, 407:242--248, 2005.

\bibitem{DaubLag:LAA92}
Ingrid Daubechies and Jeffrey~C. Lagarias.
\newblock Sets of matrices all infinite products of which converge.
\newblock {\em Linear Algebra Appl.}, 161:227--263, 1992.

\bibitem{DaubLag:LAA01}
Ingrid Daubechies and Jeffrey~C. Lagarias.
\newblock Corrigendum/addendum to \cite{DaubLag:LAA92}.
\newblock {\em Linear Algebra Appl.}, 327(1-3):69--83, 2001.

\bibitem{aldric-mean}
Aldric Degorre, Laurent Doyen, Raffaella Gentilini, Jean{-}Fran{\c{c}}ois
  Raskin, and Szymon Torunczyk.
\newblock Energy and mean-payoff games with imperfect information.
\newblock In {\em Proc.~CSL}, LNCS 6247, pages 260--274. Springer, 2010.

\bibitem{EhrenfeuchtMeanPayoff}
Andrzej Ehrenfeucht and Jan Mycielski.
\newblock Positional strategies for mean payoff games.
\newblock {\em International Journal of Game Theory}, 8(2):109--113, 1979.

\bibitem{Gurv:LAA95}
Leonid Gurvits.
\newblock Stability of discrete linear inclusion.
\newblock {\em Linear Algebra Appl.}, 231:47--85, 1995.

\bibitem{HJ2:e}
Roger~A. Horn and Charles~R. Johnson.
\newblock {\em Matrix Analysis}.
\newblock Cambridge University Press, 2013.

\bibitem{Jungers:09}
Rapha{\"e}l Jungers.
\newblock {\em The Joint Spectral Radius: Theory and Applications}.
\newblock LNCIS 385. Springer, 2009.

\bibitem{Koz:IITP13}
Victor Kozyakin.
\newblock An annotated bibliography on convergence of matrix products and the
  theory of joint/generalized spectral radius.
\newblock Preprint, Institute for Information Transmission Problems, Moscow,
  2013.
\newblock \url{http://dx.doi.org/10.13140/2.1.4257.5040}.

\bibitem{NesPro:SIAMJMAA13}
Yurii Nesterov and Vladimir~{Yu.} Protasov.
\newblock Optimizing the spectral radius.
\newblock {\em SIAM J. Matrix Anal. A.}, 34(3):999--1013, 2013.

\bibitem{Prot:MP15}
Vladimir~{Yu}. Protasov.
\newblock Spectral simplex method.
\newblock {\em Math. Program.}, pages 1--27, 2015.

\bibitem{RotaStr:IM60}
Gian-Carlo Rota and Gilbert Strang.
\newblock A note on the joint spectral radius.
\newblock {\em Nederl. Akad. Wetensch. Proc. Ser. A 63 = Indag. Math.},
  22:379--381, 1960.

\bibitem{staigerentropy}
Ludwig Staiger.
\newblock Entropy of finite-state omega-languages.
\newblock {\em Probl. Control Inform.}, 14(5):383--392, 1985.

\bibitem{Theys:PhD05}
Jacques Theys.
\newblock {\em Joint Spectral Radius: {T}heory and Approximations}.
\newblock PhD thesis, Universit{\'e} Catholique de Louvain, 2005.

\bibitem{Blondel-TsitsiklisB97}
John~N. Tsitsiklis and Vincent~D. Blondel.
\newblock The {L}yapunov exponent and joint spectral radius of pairs of
  matrices are hard --- when not impossible --- to compute and to approximate.
\newblock {\em {Math. Control Signals Systems}}, 10(1):31--40, 1997.

\bibitem{von1947theory}
John von Neumann and Oskar Morgenstern.
\newblock {\em Theory of Games and Economic Behavior}.
\newblock Princeton University Press, 1947.

\bibitem{Weih2000}
Klaus Weihrauch.
\newblock {\em Computable Analysis}.
\newblock Springer, 2000.

\bibitem{ZwickMeanPayoff}
Uri Zwick and Mike Paterson.
\newblock The complexity of mean payoff games on graphs.
\newblock {\em Theor. Comput. Sci.}, 158(1--2):343--359, 1996.

\end{thebibliography}

\newpage
\appendix
\theoremstyle{plain}%
\newtheorem{applemma}{Lemma}[section]
\newcommand{\IRU}{\mathop{IRU}}
\newcommand{\Vector}[1]{\left[
\begin{array}{c}
#1
\end{array}
\right]}
\section{Proofs}
\subsection*{Proof of Lemma~\ref{lem:iru}}
\begin{proof}
\begin{enumerate}[(i)]
\item
Let $\setA_i$ be the set of admissible $i$-th rows in $\setA$,
\[
\setR_i=\left\{\left[\sum_{k=1}^n a_k b_{kj}\right]_{1\leq j\leq n}\MID a\in \setA_i\right\},
\]
and $\setR$ be the IRU-set made from sets $\setR_i$. One has that $\setA B=\setR$:
\begin{itemize}
\item if $M\in\setR$ then let $a^{(i)}\in\setA_i$ be such that the $i$-th row of $M$ is $\left[\sum_{k=1}^n a^{(i)}_k b_{kj}\right]_{1\leq j\leq n}$, then $M=AB$ where $A$ is the matrix made with rows $a_i$;
\item conversely, if $A\in\setA$ and $a^{(i)}$ is the $i$-th row of $A$, then the $i$-th row of $AB$ equals  $\left[\sum_{k=1}^n a^{(i)}_k b_{kj}\right]_{1\leq j\leq n}$ and belongs to $\setR_i$.
\end{itemize}
\item
The easy direction is $\subseteq$. Let $M$ be a matrix of $\conv(\setA)$. Then, there exist matrices $M_1,\dots,M_k\in\setA$ and real numbers $\lambda_1,\dots,\lambda_k$ such that
\[
M=\sum_{i=1}^k \lambda_iM_i.
\]
Let $j$ be an integer in $\{1,\dots,n\}$. For all $i\in\{1,\dots,n\}$, there exists a vector $v_i\in \setA_j$ such that row $j$ of $M_i$ is $v_i$. Then, row $j$ of $M$ being $\sum_{i=1}^k \lambda_iv_i$, it belongs to $\conv{\setA_j}$.

For the direction $\supseteq$, let $M$ be a matrix of the IRU-set formed by $\conv(\setA_1),\dots,\conv(\setA_n)$. Let $u_1,\dots,u_n$ be the rows of the matrix $M$. By definition of $M$, there are integers $k_i$ for $i\in\{1,\dots,n\}$, real numbers $\lambda^i_j\in[0,1]$ and vectors $v^i_j\in \setA_i$ for $i\in\{1,\dots,n\}$ and $j\in\{1,\dots,k_i\}$ such that
\[
u_i=\sum_{j=1}^{k_i} \lambda^i_j v^i_j
\text{ and }
\sum_{j=1}^{k_i} \lambda^i_j = 1.
\]
Then, for all $i\in\{1,\dots,n\}$, one has:
$$
u_i=\sum_{j_i=1}^{k_i} \lambda^i_{j_i} v^i_{j_i}
= \left(\prod_{l=1}^{i-1}  \sum_{j_l=1}^{k_l} \lambda^l_{j_l}  \right) \left( \sum_{j_i=1}^{k_i} \lambda^i_{j_i} v^i_{j_i} \right)  \left(\prod_{l=i+1}^{n}  \sum_{j_l=1}^{k_l} \lambda^l_{j_l}  \right)\\
= \sum_{j_1=1}^{k_1}\cdots\sum_{j_n=1}^{k_n} \left(\prod_{l=1}^n \lambda^l_{j_l}\right) v^i_{j_i}.
$$
Hence
\[
M=\Vector{u_1\\\vdots\\u_n}
= \sum_{j_1=1}^{k_1}\cdots\sum_{j_n=1}^{k_n} \left(\prod_{l=1}^n \lambda^l_{j_l}\right) \Vector{{v^1_{j_1}}\\\vdots\\{v^n_{j_n}}},
\]
each matrix in the sum being in $\setA$. The proof is finished stating that
\[
\sum_{j_1=1}^{k_1}\cdots\sum_{j_n=1}^{k_n} \prod_{l=1}^n \lambda^l_{j_l}= \prod_{l=1}^{n} \sum_{j_l=1}^{k_l} \lambda^l_{j_l} = 1.
\]
\item Immediate from the characterization of compact sets (of finite dimension) as bounded and closed.	\qedhere
\end{enumerate}
\end{proof}

\subsection*{Proof of Lemmas~\ref{l:sim} and \ref{l:sim-nonneg}}

In both lemmas, we announce that it is possible to reduce the halting problem of a 2-counter Minsky machine (2CMM) to the threshold problem for MMG.

So let us first remind the reader about 2CMMs. Such a machine can be defined as a set of instructions, indexed by a finite set of states $Q$, with $q_0\in Q$ marked as initial, operating on two non-negative integer counters $x$ and $y$. There are three types of instructions ($q_{\dots}$ are states in $Q$, and $c$ is either $x$ or $y$):
\begin{enumerate}
 \item $q_i$: increment $c$ then execute $q_j$;
 \item $q_i$: if $c=0$ execute $q_j$, else decrement $c$ then execute $q_k$;
 \item $q_i$: stop.
\end{enumerate}
The computation starts from instruction $q_0$, executing it and thus triggering a sequence of instructions, which may be finite, if it eventually reaches a \emph{stop} instruction, or otherwise infinite. Whether or not the execution will be finite is undecidable.

Obviously, both reductions consist in encoding any 2CMM into an MMG, the payoffs of which depend on whether the machine halts or not.


Now let us describe the encoding used in Lemma~\ref{l:sim}. Here the 2CMM is translated into two  sets $\setA$ and $\setE$ of square matrices of dimension $|Q|+5$. States of the 2CMM (discrete location and counter values) are encoded, along with some other information, as (row) vectors of the space on which these matrices operate. The $|Q|$ first coordinates of such a vector, labelled with $q_0,\dots,q_{|Q|-1}$, take a non-zero value only for the current state of the simulation. The two next coordinates $x$ and $y$ represent the two counters. Finally, there are three additional coordinates: $One, E$ and $Neg$ (the role of which will be explained later on).

Eve's matrices allow her to simulate the machine execution (as long as it goes on). The set $\setE$ consists in exactly one matrix per transition of the 2CMM (warning: instruction 2 consists in two different transitions, depending on the test $c=0$, while instruction 3 consists in no transition, i.e., a state with \emph{stop} instruction is a deadlock state). For the sake of presentation, we describe them below as sets of assignments of variables, but it is easy to see that all assignments actually are linear operations:
\begin{itemize}
 \item matrices $I_{qq'c}$ (as \emph{Increment}):
 $q := q - \mathop{One}$; $q' := q' + \mathop{One}$; $c := c + \mathop{One}$;
  \item matrices $K_{qq'c}$ (as \emph{Keep} current counter value):
 $q := q - \mathop{One}$; $q' := q' + \mathop{One}$; $c := -c$;
 \item matrices $D_{qq'c}$ (as \emph{Decrement}):  $q := q - \mathop{One}$; $q' := q' + \mathop{One}$; $c := c - \mathop{One}$.
\end{itemize}
Notice that matrix $K_{qq'c}$ should be normally applied when $c=0$, and thus the operation $c=-c$ does not harm. On the contrary, if it is applied illegally, for a positive counter value, then it results in a negative $c$.

Adam has five kinds of matrices, which he can use to detect whenever Eve does not faithfully simulate the machine, and then punish her by forcing a payoff of $0$. Here is the set  $\setA$:
\begin{itemize}
 \item the matrix   $\mathop{Init}$ (initialize the 2CMM):
 $q_0 := E$; $q_{i\neq 0} := 0$; $x := 0$; $y := 0$; $\mathop{One} := E$; $\mathop{Neg} := 0$
 \item the identity matrix  $\mathop{Id}$ (do nothing and just let Eve continue playing);
 \item the matrices $F_c$ (\emph{flash} and take a picture of coordinate $c$) for $c$ corresponding to a state or a counter:
 $\mathop{Neg} := c$;
 \item the matrix $A$ (\emph{adjust} the value of $\mathop{Neg}$):
 $\mathop{Neg} := \mathop{Neg} + \mathop{One}$;
 \item the matrix $P$ (\emph{punish} Eve by assigning $0$ to $E$):
 $E = E + \mathop{Neg}$.
\end{itemize}

Now, in order to prove Lemma~\ref{l:sim}, it suffices to prove the two following sublemmas:
\begin{applemma}\label{l:sim-non-halting}
 The MMG obtained by the translation above from a non-halting 2CMM is determined with value 1, i.e.~it has the following properties:
 \begin{enumerate}
  \item there exists a strategy of Adam $\sigma_0$ such that for any strategy $\tau$ of Eve, $P(\sigma_0, \tau)\leq 1$;
  \item there exists a strategy of Eve $\tau_0$ such that for any strategy $\sigma$ of Adam, $P(\sigma,\tau_0)\geq 1$.
 \end{enumerate}
\end{applemma}

\begin{applemma}\label{l:sim-halting}
 The MMG obtained by the translation above from a halting 2CMM is determined with value 0, i.e.~it has the following properties:
 \begin{enumerate}
  \item there exists a strategy of Adam $\sigma_0$ such that for any strategy $\tau$ of Eve, $P(\sigma_0, \tau)\leq 0$; \item there exists a strategy of Eve $\tau_0$ such that for any strategy $\sigma$ of Adam, $P(\sigma,\tau_0)\geq 0$.
 \end{enumerate}
\end{applemma}

\begin{proof}[Proof sketch of Lemma~\ref{l:sim-non-halting}]
\begin{enumerate}
 \item  Let the strategy $\sigma_0$ consist in always playing identity.
     Then $\max_\tau P(\sigma_0,\tau)=\hat\rho(\setE)$. It is easy to see
     that applying a matrix of $\setE$ to a vector only changes the value
     of coordinates labelled by a state or by a counter, and only modifies
     it by adding or removing the value of coordinate $\mathop{One}$ (its
     value is left unchanged). Thus, for any vector $v$ and any
     sequence of matrices of $\setE$: $E_1, \dots, E_n$, all coordinates of
     vector $vE_1\cdots E_n$ are bounded in absolute value by $n\cdot
     v_{\mathop{One}}$, which means that $\|E_1\cdots E_n\|\leq k$, and
     thus $\hat\rho(\setE) \leq 1$.

\item
 Assume $\tau_0$ is as follows: Eve always stores in a variable $t$ the last time Adam played $\mathop{Init}$ (initially $t$ = 0). Then at turn $i$, she plays the matrix that corresponds to the $(i-t)$-th transition of the execution of the 2CMM.

We fix a non-negative vector
\[
v_0 = (q_0=1, q_{i\neq 0}=0, x=0, y=0, E=1,
\mathop{One}=1, \mathop{Neg}=0)
\]
and prove by induction the following invariant on the vector  $v_n =
v_0A_1E_1A_2E_2\cdots A_nE_n$:
\begin{itemize}
 \item all coordinates of $v_n$ are non-negative;
 \item ${v_n}^E \geq 1$.
\end{itemize}

Indeed applying a matrix of $\setA$ while respecting the rules of the 2CMM ensures that state and counter coordinates remain non-negative, while $\mathop{One}, E$ and $\mathop{Neg}$ remain unchanged.

On the other hand Adam's matrices are all non-negative, implying the first bullet of the invariant, and can only modify $E$ by adding the value of $\mathop{Neg}$, which, by the invariant, was non-negative at the previous step.

The above proves the invariant which implies the second item of the lemma. \qedhere

\end{enumerate}

\end{proof}	
\begin{proof}[Proof sketch of Lemma~\ref{l:sim-halting}]
\begin{enumerate}
 \item For $\sigma_0$, we consider the following strategy:
 \begin{itemize}
  \item first play $\mathop{Init}$: Adam initializes a simulation of the
      2CMM in his private memory in the form of a vector $v_0 = (q_0=1,
      q_{i\neq 0}=0, x=0, y=0, E=1, \mathop{One}=1, \mathop{Neg}=0)$, on
      which all subsequent matrices will be applied (yielding
      $v_1,v_2,\dots$);
  \item then play $\mathop{Id}$ as long as Eve plays valid transitions of the 2CMM;
  \item play $F_c$ as soon as Eve plays an invalid move (if Eve lied on a counter value, then $c$ is the name of this counter; if Eve lied on current state, then $c$ is the name of this state; in both cases this corresponds to a negative coordinate);
  \item play $A$ until ${v_n}^{\mathop{Neg}} = -1$;
  \item finally play $P$ (nulling $E$) and a last time $\mathop{Init}$ (nulling the whole vector).
 \end{itemize}

 Explanation: the invariant from Lemma~\ref{l:sim-non-halting} holds as
 long as Eve simulates the 2CMM. When she stops simulating, the value of
 $E$ is still $1$ but some coordinate $c$ is negative. The ending sequence
 $F_cA\cdots AP\mathop{Init}$ forces the final vector to be $0$, no move
 of Eve can then prevent this from happening, as she cannot modify
 $\mathop{One}, E$ or $\mathop{Neg}$.

 Now remark that, for any vector $v$ such that  $v^E=1$, it holds that $v\cdot \mathop{Init} = v_0$ and thus $v\cdot\Omega = 0$ where $\Omega$ is the product matrix for the whole play until Adam plays a last time $\mathop{Init}$. This proves that as much yields for any initial vector, thus $\Omega = 0$ and therefore $P(\sigma_0, \tau) = 0$.

\item  The payoff function of an MMG is always non-negative. \qedhere
\end{enumerate}
\end{proof}	

\begin{proof}[Proof of Lemma~\ref{l:sim}]
It directly follows from  Lemmas~\ref{l:sim-non-halting} and \ref{l:sim-halting}.
\end{proof}	
\begin{proof}[Proof of Lemma~\ref{l:sim-nonneg}]

\emph{General idea:} Here, in order to use only non-negative matrices,  we introduce a slightly different construction. Indeed, previous encoding relied on the fact that when Eve cheats, a negative coordinate appears that Adam can use to punish her (by nulling the product matrix with clever additions using the negative integers). Now there is no hope to create a negative element in the product matrix (since the matrices of the MMG are non-negative), so the punishment will be less drastic: Adam will just try to obtain a product that grows more slowly than the product of matrices corresponding to a faithful simulation. For this Adam needs to reset the game infinitely often, as to force Eve to cheat as many times, within a bounded horizon.

The encoding uses the following idea: a counter of value $k$ is encoded at time $n$ by a coordinate of value $2^{n+k}$. This way, a counter decrement consists in keeping its coordinate value unchanged, while a counter increment consists in multiplying its coordinate by 4. A counter stalling at a given step still sees its coordinate multiplied by 2.

\emph{The vector space:}
Matrices are square of dimension $|Q|+4$. They act on vectors such that their $|Q|$ first coordinates represent the states of the 2CMM (positive value only in the coordinate corresponding to the active state of the simulation) and the four other coordinates $x_+, x_-, y_+, y_-$, are the two counters and their opposite (i.e.~their value will be $2^{n-c}$ instead of $2^{n+c}$).

\emph{Eve's matrices:}
In this game too, Eve tries to faithfully simulate the 2CMM. Her matrices are the following:
\begin{itemize}
 \item matrices $I_{qq'x}$:
 $q' := 2 q$; $q := 0$; $x_+ := 4 x_+$;  $y_+ := 2 y_+$;  $y_- := 2 y_-$;
  \item matrices $K_{qq'x}$:
 $q' := 2 q$; $q := 0$;  $(x_+,x_-) := 2 (x_-, x_+)$; $(y_+, y_-) := 2 (y_-, y_+)$;
\item matrices $D_{qq'x}$:
 $q' := 2 q$; $q := 0$; $x_- := 4 x_-$;  $y_+ := 2 y_+$;  $y_- := 2 y_-$;
 \item matrices $I_{qq'y}$, $K_{qq'y}$ and $D_{qq'y}$ are defined likewise.
\end{itemize}

Notice that the coordinate inversion of the previous construction, for the case of a successful $x=0$ test, is now translated as a coordinate swap between $x_+$ and $x_-$. Thus when Eve cheats on a counter value, be it one way or the other, $x_+$ has a value smaller than $2^n$.

\emph{Adam's matrices:}
\begin{itemize}
 \item matrix $\mathop{Id}$;
 \item matrices $P_x$ (and $P_y$): $q_0 := x_+$; $q_{i\neq 0} := 0$; $x_-=x_+$; $y_+=x_+$; $y_-=x_+$;
 \item matrices $P_q$: let $s=\sum_{q\in Q} q$ then $q_0 := s$; $q_{i\neq 0} := 0$; $x_+=s$;$x_-=s$; $y_+=s$; $y_-=s$.
\end{itemize}
Both $P_{\{x,y\}}$ and $P_q$ reset the simulation, forcing the copied value as the new norm for the product matrix.

Adam's strategy consists in playing $\mathop{Id}$ most of the time; playing $P_q$ whenever Eve cheats on the state, implying a null product; and playing $P_c$ whenever Eve cheats on a counter $c\in\{x,y\}$ value, implying a factor of norm $\leq 2^{f - 1}$ since the last time when Adam played a $P$ matrix (factor of length $f$).

Since Eve needs to cheat with a positive frequency if the run of the 2CMM is finite, then the final payoff will be  $<2$ (from the product of such factors, which are of bounded length).

If the run is infinite, whether Adam plays $\mathop{Id}$ or a $P$, there will be some coordinate that remains of magnitude $\geq 2^n$.

Details of the proof are similar to those of Lemma~\ref{l:sim}. We prove determinacy, with a value of 2, in the case when the 2CMM has an infinite run. For the other case we prove that Adam can ensure a payoff $<2$, but determinacy remains an open problem.
\end{proof}

\subsection*{Proof of Lemmas~\ref{lemma:f1}--\ref{L:equiv}}

\begin{proof}[Proof of Lemma~\ref{lemma:f1}]
Assume $D = \{d_1,\ldots, d_M\}$ and $T = \{t_1,\ldots, t_N\}$. Given
arbitrary strategies  $(\sigma,\tau)$ for the two players in the EG, let us
represent the set of all compatible plays as a forest. Its nodes  are labeled
by elements of $D$ on even levels and elements of $T$ on odd levels, and its
edges are labeled by symbols in $\Sigma$. The label of a node $q$ is denoted
$\ell(q)$; the sequence of labels on the path reaching $q$ from the
appropriate root in the forest is referred to as its address $\alpha(q)$. The
forest $\forest$ is  defined inductively as follows:
\begin{itemize}
	\item $\forest$ has $M$ root nodes labeled by $d_1,\dots,d_M$;
	\item all the outgoing edges of a node $q$ labeled $d\in D$ carry the symbol $a = \sigma(\alpha(q))$
	and the sons of the  node $q$ correspond to (and are labeled by) the elements of $\{t \mid (d,a,t) \in \Delta \}$;
	\item all the outgoing edges of a node $q$ labeled $t\in T$ carry the symbol $b = \tau(\alpha(q))$
	and the sons of the  node $q$ correspond to (and are labeled by) the elements of $\{d \mid (t,a,d) \in \Delta \}$.
\end{itemize}
The payoff of the EG can be characterized in terms of the growth rate of this
forest:
\[
P(\sigma,\tau)=\limsup_{n \to \infty} |\forest_{2n}|^{1/n},
\]
where $\forest_k$ denotes the set of nodes of $\forest$ at the level $k$.
Indeed $L(\sigma,\tau)$ is the set of labels of infinite paths of $\forest$,
hence $\pref(L(\sigma,\tau))$ is the set of addresses of nodes in $\forest$
(we use the fact that our strategies are required to be non-blocking). To
words of length $4n$ in $\pref(L(\sigma,\tau))$ correspond addresses of nodes
of level $2n$, and thus
\[
\limsup_{n \to \infty} |\pref_{4n}(L(\sigma,\tau))|^{1/n}=\limsup_{n \to \infty} |\forest_{2n}|^{1/n}
\] as required.

Let us characterize the number of nodes $|\forest_{2n}|$ in terms of
matrices. Let the vector $x^{(n)}=(x^{(n)}_1,\dots,x^{(n)}_j)$ be such that
$x^{(n)}_i$ is the number of nodes labeled by $d_i$ on $2n$-th level of
$\forest$; similarly let   $y^{(n)}=(y^{(n)}_1,\dots,y^{(n)}_N)$ be such that
$y^{(n)}_j$ is the number of nodes labeled by $t_j$ on $(2n+1)$-th level of
$\forest$. To relate $y^{(n)}$ to $x^{(n)}$ we observe that
\[
y^{(n)}_j=\sum_{i=1}^M\sum_{a\in\Sigma}\left|\big\{q\in \forest_{2n}\MID  \ell(q)=d_i \land \sigma(\alpha(q))=a\big\}\right|c_{ia,j}.
\]
Indeed, every node on level $2n$ with label $d_i$ and action $a$ generates on
the next level a node with label $t_j$ whenever $c_{ia,j}=1$. Summing up on
all $i$, $a$ and $q$ we obtain the quantity  $y^{(n)}_j$.
The expression for $y$ can be rewritten as
\begin{equation}\label{E:musum}
y^{(n)}_j=\sum_{i=1}^M x^{(n)}_i\sum_{a\in\Sigma}\mu_{ia}c_{ia,j}
\end{equation}
with $\mu_{ia}^{(n)}=\left|\big\{q\in \forest_{2n}\MID\ell(q)=d_i \land
\sigma(\alpha(q))=a\big\}\right|/x^{(n)}_i$ (whenever $x^{(n)}_i=0$,
coefficients $\mu_{ia}^{(n)}$ can be chosen arbitrarily, only respecting
conditions \eqref{E:mu} below). Intuitively, $\mu_{ia}^{(n)}$ is the
proportion among the states $d_i$ on level $2n$, of those for which Despot
takes the action $a$.
In matrix form \eqref{E:musum} can be rewritten as $y^{(n)}=x^{(n)} A_n$ with
$A_{n,ij}=\sum_{a\in\Sigma}\mu_{ia}^{(n)}c_{ia,j}$. We notice that
\begin{equation}\label{E:mu}
	\mu_{ia}^{(n)}\ge 0 \text{ and } \sum_{a\in\Sigma}\mu_{ia}^{(n)}=1,
\end{equation}
thus $i$-th row of $A_n$ belongs to $\conv(\setA_i)$, hence
$A_n\in\conv(\setA)$. Similarly, $x^{(n+1)}=y^{(n)}E_n$ for some
$E_n\in\conv(\setE)$. Initially $x^{(0)}=(1,\dots,1)$, and clearly
$|\forest_n|=x^{(n)}\cdot (1,\dots,1)^\transpose$, hence
\[
|\forest_{2n}|= (1,\dots,1)A_0E_0A_1E_1\cdots A_{n-1}E_{n-1} (1,\dots,1)^\transpose=\|A_0E_0A_1E_1\cdots A_{n-1}E_{n-1}\|.
\]
Taking  in the MMG over $(\conv(\setA)),\conv(\setE))$ the strategies
$\varsigma$ and $\theta$, which choose matrices $A_0,E_0,A_1,E_1,\dots$ we
obtain the required:
\[
P_{\text{EG}}(\sigma,\tau)=\limsup_{n \to \infty} |\forest_{2n}|^{1/n}=\limsup_{n \to \infty} \|A_0E_0A_1E_1\cdots A_{n-1}E_{n-1}\|^{1/n}=P_{\text{MMG}}(\varsigma,\theta).
\]

It is easy to see that for positional $\sigma$ our construction gives
$A_n=A_\sigma$ for all $n$.
\end{proof}

\begin{proof}[Proof of Lemma~\ref{L:1}]
As stated in \cite[Corollary~8.1.29]{HJ2:e}, for any nonnegative matrix $A$ and $u>0$%
\begin{equation}\label{E:horn29}
\alpha u\le Au\le \beta u \Rightarrow \alpha\le \rho(A)\le \beta,
\end{equation}
our statement (i) is now immediate. Let us prove the three remaining assertions.
\begin{enumerate}[(i)]\addtocounter{enumi}{1}
\item Let  $Au\le\rho u$ for $u>0$ with $A>0$ and $Au\neq\rho u$.
Then at least one coordinate of the vector $Au-\rho u\le 0$ is strictly negative.
Therefore the condition $A>0$ implies strict negativity of all coordinates of the vector $A(Au-\rho u)$.
Then there exists $\varepsilon>0$ such that $A(Au-\rho u)\le -\varepsilon u$ and therefore $A^{2}u=A(Au-\rho u)+\rho Au\le (\rho^{2}-\varepsilon)u$.
Then, by \eqref{E:horn29}, we get   $\rho(A^{2})\le\rho^{2}-\varepsilon$, and thus $\rho(A)\le\sqrt{\rho^{2}-\varepsilon}<\rho$, q.e.d.

\item The condition $Au\ge\rho u$ with non-zero $u\ge0$ implies $A^{n}u\ge \rho^{n}u$ for any $n\ge1$.
Then $\|A^{n}\|\cdot\|u\|\ge\|A^{n}u\|\ge \rho^{n}\|u\|$.
Therefore $\|A^{n}\|\ge \rho^{n}$, and by Gelfand's formula $\rho(A)\ge\rho$, q.e.d.

\item Now let $A>0$ and $Au\neq\rho u$.
Then at least one coordinate of the vector $Au-\rho u\ge 0$ is strictly positive.
Therefore the condition $A>0$ implies strict positivity of all the coordinates of the vector $A(Au-\rho u)$.
Then there exists $\varepsilon>0$ such that $A(Au-\rho u)\ge \varepsilon u$ and therefore $A^{2}u=A(Au-\rho u)+\rho Au\ge (\rho^{2}+\varepsilon)u$.
This, by (iii) applied to the matrix $A^{2}$, implies  $\rho(A^{2})\ge\rho^{2}+\varepsilon$, and thus $\rho(A)\ge\sqrt{\rho^{2}+\varepsilon}>\rho$, q.e.d. \qedhere
\end{enumerate}
\end{proof}

\begin{proof}[Proof of Lemma~\ref{L:alternative}]
To prove (i), we represent the vectors $u$ and $v$ in coordinate form:
\[
u=(u_{1},u_{2},\ldots,u_{M})^{\transpose},\quad
v=(v_{1},v_{2},\ldots,v_{N})^{\transpose}.
\]
Suppose that for some matrix $A=(a_{ij})\in\setA$ the inequality $Au\ge v$
fails. Then
\[
a_{i1}u_{1}+a_{i2}u_{2}+\cdots+a_{iM}u_{M}<v_{i}
\]
for some $i\in\{1,2,\ldots,N\}$; we may assume $i=1$ without loss of
generality. In this case, the matrix
\[
\bar{A}=\left[\begin{array}{cccccc}
a_{11}&a_{12}&\cdots&a_{1M}\\
\tilde{a}_{21}&\tilde{a}_{22}&\cdots&\tilde{a}_{2M}\\
\cdots&\cdots&\cdots&\cdots\\
\tilde{a}_{N1}&\tilde{a}_{N2}&\cdots&\tilde{a}_{NM}
\end{array}\right],
\]
obtained from the matrix $\tilde{A}=(\tilde{a}_{ij})$
replacing the first row by $a_{1}=[a_{11},a_{12},\ldots,a_{1M}]$, yields the
inequalities
$ a_{11}u_{1}+a_{12}u_{2}+\cdots+a_{1M}u_{M}<v_{1}$; and
$\tilde{a}_{i1}u_{1}+\tilde{a}_{i2}u_{2}+\cdots+\tilde{a}_{iM}u_{M}=v_{i}$
for $i=2,3,\ldots,N. $ Consequently, $\bar{A}u\le v$ and $\bar{A}u\neq v$,
which completes the proof of the first statement of the lemma.
The proof of statement (ii) is similar.
\end{proof}

\begin{proof}[Proof of Corollary~\ref{C:minimaxconv}]
We denote  $V=\minimax(\setA,\setB)$ and
$V'=\minimax(\conv(\setA),\conv(\setB))$. Then
\[
V'\stackrel1=\min_{A\in\conv(\setA)}\max_{B\in\conv(\setB)}\rho(BA)\stackrel2\le
\min_{A\in\setA}\max_{B\in\conv(\setB)}\rho(BA)\stackrel3=\min_{A\in\setA}\max_{B\in\setB}\rho(BA)=V,
\]
where 1 follows from the equality $\rho(AB)=\rho(BA)$, 2 from the inclusion
$\setA\subseteq \conv(\setA)$, 3 from Lemma~\ref{lem:iru} and
equalities~\eqref{E:minmax-conv}. Symmetrically,
\[
V'=\max_{B\in\conv(\setB)}\min_{A\in\conv(\setA)}\rho(AB)\ge
\max_{B\in\setB}\min_{A\in\conv(\setA)}\rho(AB)=\max_{B\in\setB}\min_{A\in\setA}\rho(AB)=V.\qedhere
\]
\end{proof}

\begin{proof}[Proof of Lemma~\ref{lemma:f2}]
	Let $\sigma'$ and $\tau'$  be arbitrary strategies in the EG, then   by Lemma~\ref{lemma:f1}  for the strategy pair $(\sigma',\tau)$ there is a corresponding pair $(\varsigma',E_\tau)$ with some strategy $\varsigma'$ having the same value in the MMG.  Symmetrically for the pair $(\sigma,\tau')$ there is a corresponding pair $(A_\sigma,\theta')$. We have:
	\[
	P(\sigma',\tau) = P(\varsigma',E_\tau) \le P(A_\sigma,E_\tau) = P(\sigma,\tau) =  P(A_\sigma,E_\tau) \le P(A_\sigma,\theta') = P(\sigma,\tau'),
	\]
	where the equalities come from Lemma~\ref{lemma:f1} and the inequalities from the optimality of $E_\tau$ and $A_\sigma$, respectively. Thus $\sigma$ and $\tau$ are optimal.
\end{proof}

We will need the following result in order to prove  Lemma~\ref{L:equiv}.
\begin{applemma}\label{L:mainmin}
Let $\setA$ be a compact IRU-set of positive $(N\times N)$-matrices.
\begin{enumerate}[\rm(i)]
\item If $\tilde{A}\in\setA$ is a matrix satisfying $\rho(\tilde{A}) =
    \check{\rho}(\setA)$ and $\tilde{v}$ is its positive eigenvector corresponding
    to the eigenvalue $\rho(\tilde{A})$, then $A\tilde{v}\ge
    \check{\rho}(\setA)\tilde{v}$ for all $A\in\setA$.
\item If $\tilde{A}\in\setA$ is a matrix satisfying $\rho(\tilde{A})
    = \hat{\rho}(\setA)$ and $\tilde{v}$ is its positive eigenvector corresponding
    to the eigenvalue $\rho(\tilde{A})$, then $A\tilde{v}\le
    \hat{\rho}(\setA)\tilde{v}$ for all $A\in\setA$.
\end{enumerate}
\end{applemma}
\begin{proof}[Proof of Lemma~\ref{L:mainmin}]
To prove (i) let us note that
$\tilde{A}\tilde{v}=\check{\rho}(\setA)\tilde{v}$. Then by
Lemma~\ref{L:alternative}(i) either
$A\tilde{v}\ge\check{\rho}(\setA)\tilde{v}$ for all $A\in\setA$ or there
exists a matrix $\bar{A}\in\setA$ such that $\bar{A}\tilde{v}\le
\check{\rho}(\setA)\tilde{v}$ and $\bar{A}\tilde{v}\neq
\check{\rho}(\setA)\tilde{v}$. In the latter case, by Lemma~\ref{L:1}
 the inequality $\rho(\bar{A})<\check{\rho}(\setA)$ would hold,  which
contradicts to the definition of $\check{\rho}(\setA)$. Hence, the inequality
$A\tilde{v}\ge\check{\rho}(\setA)\tilde{v}$ holds for all $A\in\setA$, q.e.d.
Assertion (ii) is proved similarly.
\end{proof}

\begin{proof}[Proof of Lemma~\ref{L:equiv}]
  For positive matrices, implications $\Leftarrow$ follow from Lemma~\ref{L:1}. As for $\Rightarrow$, it suffices to take $v$ the eigenvector (corresponding to the spectral radius) of the matrix $\tilde{A}\in\setA$ with the largest (smallest) spectral radius, and to apply Lemma~\ref{L:mainmin}.

As for non-negative matrices, we have four implications to prove:
\begin{description}
  \item[\eqref{E:equiv:lt}, $\Rightarrow$] Denote, for any $\varepsilon>0$,
      $\setA_{\varepsilon}=\{A+\varepsilon\1\Mid A\in\setA\}$. If
      $\hat{\rho}(\setA)<\alpha$ then due to compactness of the set $\setA$
      there exists $\varepsilon>0$ such that
      $\hat{\rho}(\setA_{\varepsilon})
      =\hat{\rho}(\setA+\varepsilon\1)<\alpha$. Then by  \eqref{E:equiv:lt}
      (already proved for positive matrices), there exists $v>0$ such that
      $(A+\varepsilon\1)v<\alpha v$ for all $A\in\setA$. Since
      $Av\le(A+\varepsilon\1)v$, then $Av<\alpha v$ for all $A\in\setA$,
      q.e.d.
\item[\eqref{E:equiv:lt}, $\Leftarrow$] Suppose there exists $v>0$ such
    that $Av<\alpha v$ for all $A\in \setA$. Then due to compactness of the
    set $\setA$ there exists $\varepsilon>0$ such that
    $(A+\varepsilon\1)v<\alpha v$ for all $A\in\setA$. Therefore by
    \eqref{E:equiv:lt} (for positive matrices)
    $\hat{\rho}(\setA+\varepsilon\1)<\alpha$, and hence by the monotonicity
    of the spectral radius we obtain $\hat{\rho}(\setA)<\alpha$, q.e.d.
\item[\eqref{E:equiv:ge2}, $\Rightarrow$] Let
    $\check{\rho}(\setA)\ge\alpha$, then by the monotonicity of the
    spectral radius it holds that
    $\check{\rho}(\setA+\varepsilon\1)\ge\alpha$ for any $\varepsilon>0$.
    Then by  \eqref{E:equiv:ge} (for positive matrices) for any
    $\varepsilon>0$ exists a vector $v_{\varepsilon}>0$ such that
    $\|v_{\varepsilon}\|=1$ and
\begin{equation}\label{E:Ae}
(A+\varepsilon\1)v_{\varepsilon}\ge\alpha v_{\varepsilon}
\end{equation}
for all $A\in\setA$. Choose  a sequence $\varepsilon_{n}\to0$ for which the
corresponding vectors $v_{\varepsilon_{n}}$ converge to some vector $v\ge
0$ (let us point out that $\|v\|=1$ and so it is non-zero). Then passing to
the limit in \eqref{E:Ae} we obtain $Av\ge\alpha v$ for all $A\in\setA$,
q.e.d.
\item[\eqref{E:equiv:ge2}, $\Leftarrow$] Suppose there exists a non-zero
    vector $v\ge0$ such that $Av\ge\alpha v$ for all $A\in \setA$. Then by
    Lemma~\ref{L:1}, $\rho(A)\ge\alpha$ for all $A\in \setA$ and hence
    $\check{\rho}(\setA)\ge\alpha$, q.e.d.\qedhere
\end{description}
 \end{proof}

\end{document}